\documentclass[11pt,letterpaper,twoside,reqno,nosumlimits]{amsart}

\usepackage{etoolbox}

\patchcmd{\section}{\scshape}{\bfseries}{}{}
\makeatletter
\renewcommand{\@secnumfont}{\bfseries}
\makeatother

\usepackage[dvipsnames]{xcolor}
\patchcmd{\section}{\normalfont}{\normalfont\color{MidnightBlue}}{}{}
\patchcmd{\subsection}{\normalfont}{\normalfont\color{MidnightBlue}}{}{}

\makeatletter
\def\subsubsection{\@startsection{subsubsection}{3}%
\z@{.5\linespacing\@plus.7\linespacing}{-.5em}%
{\normalfont\bfseries}}
\makeatother

\usepackage{circuitikz}

\usepackage{fancyhdr}
\usepackage{amsmath,amsfonts,amsbsy,amsgen,amscd,mathrsfs,amssymb,amsthm,mathtools,tensor}

\usepackage[a4paper]{geometry}

\usepackage{overpic}

\usepackage{stmaryrd}

\usepackage{amsopn}

\usepackage{caption}
\usepackage{graphicx}
\usepackage{color}
\usepackage[bf,SL,BF]{subfigure}
\usepackage{url}
\usepackage{epstopdf}
\usepackage{pdfsync}
\usepackage[colorlinks]{hyperref}
\usepackage[all]{xypic}
\usepackage{comment}
\usepackage{mathabx}
\usepackage{upgreek}

\usepackage{mathrsfs}
\usepackage{yfonts}
\usepackage{multirow}

\hypersetup{
    linkcolor=blue,
}

\newlength{\fixboxwidth}
\usepackage{epsfig,amsbsy,graphicx,multirow}

\usepackage{ algorithm, algorithmic}
\renewcommand{\algorithmiccomment}[1]{\bgroup\hfill//~#1\egroup}

\usepackage[all,cmtip]{xy}

\usepackage{accents}

 \numberwithin{equation}{section}

\makeatletter

\renewcommand\subsubsection{\@secnumfont}{\bfseries}%
\renewcommand\subsubsection{\@startsection{subsubsection}{3}
  \z@{.5\linespacing\@plus.7\linespacing}{-.5em}%
  {\normalfont\bfseries}}
\makeatother

\def\g{\mathfrak{g}}

\def\curl{\operatorname{\bf curl}}

\def\R{\mathbb{R}}

\def\cN{\mathcal{N}}

\def\S{\mathcal{S}}

\def\Y{{\bf\mathcal{Y}}}

\def\E{\mathbb{E}}

\def\C{{\mathcal{C}}}

\def\G{\mathcal{G}}

\def\L{\mathcal{L}}

\def\s{\mathfrak{s}}

\def\H{\mathcal{H}}
\def\T{\mathbb{T}}

\def\restrict#1{\raise-.5ex\hbox{\ensuremath|}_{#1}}

\def\<{\big\langle}
\def\>{\big\rangle}

\def\diiv{\operatorname{\bf div}}

\def\Var{\operatorname{Var}}
\def\Tr{\operatorname{Tr}}

\def\Cov{\operatorname{Cov}}

\def\Hess{\operatorname{Hess}}

\def\s{\mathfrak{s}}

\definecolor{red}{rgb}{0.9, 0, 0}

\definecolor{green}{rgb}{0.0, 1.0, 0.0}

\newtheorem{Theorem}{Theorem}[section]
\newtheorem{Proposition}[Theorem]{Proposition}

\newtheorem{Remark}[Theorem]{Remark}

\makeatletter
\newcommand{\oset}[3][0ex]{%
  \mathrel{\mathop{#3}\limits^{
    \vbox to#1{\kern-2\ex@
    \hbox{$\scriptstyle#2$}\vss}}}}
\makeatother

\makeatletter
\newcommand{\uset}[3][0ex]{%
  \mathrel{\mathop{#3}\limits_{
    \vbox to#1{\kern-2\ex@
    \hbox{$\scriptstyle#2$}\vss}}}}
\makeatother

\usepackage{aurical}
\usepackage[T1]{fontenc}

\usepackage{tikz,pgflibraryshapes}
\usetikzlibrary{shadows}
\usetikzlibrary{arrows}

\usepackage{slashbox,booktabs}

\begin{document}
\title[Gaussian Process Hydrodynamics ]{Gaussian Process Hydrodynamics}

\date{\today}

\newcommand{\ho}[1]{{\color{blue} #1}}

\author{Houman Owhadi}

\thanks{Caltech,  MC 9-94, Pasadena, CA 91125, USA, owhadi@caltech.edu}

\maketitle

\begin{abstract}
We present a Gaussian Process (GP) approach (Gaussian Process Hydrodynamics, GPH) for approximating the solution of the Euler and Navier-Stokes equations. As in Smoothed Particle Hydrodynamics (SPH), GPH is a Lagrangian particle-based approach involving the tracking of a finite number of particles transported by the flow. However, these particles do not represent mollified particles of matter but carry discrete/partial information about the continuous flow. Closure is achieved by placing a divergence-free GP prior $\xi$ on the velocity field and conditioning on vorticity at particle locations. Known physics (e.g., the Richardson cascade and velocity-increments power laws) is incorporated into the GP prior through physics-informed additive kernels. This is equivalent to expressing $\xi$ as a sum of independent GPs $\xi^l$, which we call modes, acting at different scales (each mode $\xi^l$  self-activates to represent the formation of eddies at the corresponding scales). This approach leads to a quantitative analysis of the Richardson cascade through the analysis of the activation of these modes and allows us to coarse-grain turbulence in a statistical manner rather than a deterministic one. Since GPH is formulated on the vorticity equations, it does not require solving a pressure equation. By enforcing incompressibility and fluid/structure boundary conditions through the selection of the kernel, GPH requires much fewer particles than SPH. Since GPH has a natural probabilistic interpretation, numerical results come with uncertainty estimates enabling their incorporation into a UQ pipeline and the adding/removing of particles (quantas of information) in an adapted manner. The proposed approach is amenable to analysis, it inherits the complexity of state-of-the-art solvers for dense kernel matrices, and it leads to a natural definition of turbulence as information loss. Numerical experiments support the importance of selecting physics-informed kernels and illustrate the major impact of such kernels on accuracy and stability. Since the proposed approach has a Bayesian interpretation, it naturally enables data assimilation and making predictions and estimations based on mixing simulation data with experimental data.
\end{abstract}

\section{Introduction}
The Navier-Stokes (NS) equations are not only difficult to analyze \cite{fefferman2000existence}, the emergence of multiple nonlinearly coupled scales makes them hard to approximate numerically. Even from a physicist's perspective, they remain poorly understood, and we still do not have a clear definition of turbulence beyond ``the complex, chaotic motion of a fluid'' \cite{phillips2018turbulence}.
The NS equations are also hard to solve because they contain a dual description of the underlying physics that is Lagrangian in its representation of Newton's second law and Eulerian in its description of the pressure equation.
Classical methods for solving the NS equations are correspondingly divided into Eulerian (grid-based) and Lagrangian (meshfree particle-based) methods.
While Eulerian methods are efficient in handling pressure equations, they require high resolutions to handle the Lagrangian effects of the equations.
While Lagrangian methods are efficient in replicating conservation laws (e.g., entropy, momentum, energy), they require a large number of particles to handle the Eulerian aspects of the equations (e.g., solve for pressure given the position/velocities of the particles).

\subsection{Smoothed Particle Hydrodynamics}
Smoothed Particle Hydrodynamics (SPH) is a prototypical Lagrangian  meshfree particle  method  (where continuum is
assumed to be a collection of imaginary particles)  introduced in the late 1970s for  astrophysics problems \cite{lucy1977numerical,gingold1977smoothed} (see \cite{monaghan1992smoothed} for a review).
Although SPH has, by now, been  widely applied to different areas in engineering
and science (see \cite{liu2010smoothed} for an overview), including computational fluid dynamics (CFD), it suffers from the difficulties associated with Lagrangian methods and
``still requires development to
address important elements which prevent more widespread use'' \cite{vacondio2021grand}. These elements (identified as grand challenges in \cite{vacondio2021grand}) include  (1) convergence, consistency and stability, (2) boundary
conditions, (3) adaptivity, (4) coupling to other models, and (5) applicability to industry.

\subsection{Gaussian Process Hydrodynamics}
The purpose of this paper is to introduce Gaussian Process Hydrodynamics (GPH) as an information/inference-based approach to approximating the NS equations. Although numerical approximation and statistical inference may be seen as  separate subjects,
they are intimately connected through the common purpose of making estimations with partial information \cite{OwhScoSchNotAMS2019}, and
Kernel/GP methods provide a natural (and minimax optimal \cite{owhadi2019operator}) approach to computing with missing information.
 In the proposed GPH approach, flow-advected particles carry partial information on the underlying vorticity/velocity fields, and (information gap) closure is achieved by randomizing the underlying velocity field via a Gaussian Process (GP) prior with a physics-informed kernel, ensuring that incompressibility and boundary conditions are exactly satisfied, and power/scaling and energy transfer laws are satisfied in a statistical manner.
From this perspective, turbulence can be defined and quantified as information loss between the true dynamic of the NS equations and the one resulting from carrying only partial information on the underlying fields.
Although GPH has similarities with SPH, it also has several significant differences: (1) In SPH, particles represent mollified particles of matter; in GPH, particles represent discrete/partial information about the continuous flow. (2) SPH is typically formulated on the velocity and requires solving a pressure equation; GPH is formulated on the vorticity equations and Eulerian aspects (e.g., recovering the velocity field) are handled through Gaussian Process Regression. (3) By enforcing incompressibility and fluid/structure boundary conditions through the selection of the kernel, GPH requires much fewer particles. (4) By carrying variance information, GPH enables adding and removing quantas of information from the flow in an adapted manner.

While SPH recovers fields through smooth approximations of delta Dirac functions with compactly supported kernels, the focus of GPH is on the optimal recovery \cite{micchelli1977survey, owhadi2019operator} of the missing information with adapted/programmed kernels \cite{owhadi2019kernelmd}.
Its representation of the multiscale structure of the flow through regression additive kernels enables a corresponding statistical decomposition of the flow at different scales (modes), and a quantitative analysis of the Richardson cascade through the analysis of the activation of these modes \cite{owhadi2019kernelmd}.
Its focus on informing the kernel about the underlying physics and boundary conditions
opens a different strategy for solving some of the grand challenges of SPH listed above. Its probabilistic/Bayesian interpretation enables its incorporation into Uncertainty Quantification (UQ) pipelines.

\subsection{Vortex methods}
Since GPH resembles vortex methods \cite{leonard1980vortex, cottet2000vortex} (due to its formulation on the vorticity equations), it can also be interpreted as a generalization of such methods to arbitrary kernel approximations of the underlying vorticity and velocity fields based on discrete vorticity information carried by the Lagrangian particles. However, the velocity field is not recovered from the continuous vorticity field using the Biot-Savart law but from the available partial information about the continuous vorticity field using kernel (GPR) representer formulas.

\subsection{Solving PDEs as  learning problems}
There are essentially two main approaches to solving PDEs as learning problems: (1) ANN-based approaches with Physics Informed Neural Networks \cite{raissi2019physics, karniadakis2021physics} as a prototypical example, and (2) GP-based approaches with Gamblets \cite{Owhadi:2014, OwZh:2016, owhadi2017multigrid} as a prototypical example.
Although  GP-based approaches are more theoretically well-founded \cite{owhadi2019operator} and have a long history of interplays with numerical approximation \cite{OwhScoSchNotAMS2019, schafer2021sparse, SchaeferSullivanOwhadi17, yoo2019noising}, they were essentially limited to linear/quasi-linear/time-dependent PDEs and were only recently generalized to arbitrary nonlinear PDEs \cite{chen2021solving} (and to computational graphs \cite{owhadi2022computational}).

\subsection{Physics-informed kernels}
While both ANN and GP methods replace the solution of the PDE with an ANN/GP and are physics-informed by constraining/conditioning the ANN/GP to satisfy the PDE over a finite number of degrees of freedom (e.g., collocation points), GP methods can also be physics-informed through their kernels \cite{Owhadi:2014}.
The importance of employing physics/PDE-informed is well understood in numerical approximation/homogenization using Darcy's elliptic PDE $-\diiv(a\nabla)$ (with rough conductivity $a$) as a prototypical example. Indeed, while employing a smooth kernel may lead to arbitrary bad convergence \cite{BaOs:2000}, employing a physics-informed kernel ensures an optimal rate of convergence \cite{Owhadi:2014}.
While \cite{Owhadi:2014} proposed to identify such kernels by
 filtering white noise through the solution operator of the PDE (i.e., replacing the right-hand side/source term with white noise and conditioning the resulting randomized solution on a finite number of linear measurements), this approach is not practical for nonlinear PDEs since the resulting solution is not a GP.

The approach proposed in this paper is to select a physics-informed kernel by programming the kernel \cite{owhadi2019kernelmd} to satisfy: (1) the divergence-free condition of the velocity field, (2) boundary conditions (3) statistical power laws, and (4) the Richardson cascade of turbulence.

\subsection{Outline of the article}
The remainder of the article is organized as follows.
Sec.~\ref{secsetup} and \ref{secGPH} introduce GPH in the setting of the vorticity formulation of the forced Navier-Stokes equations.
Sec.~\ref{secdivfreegp} describes representer formulas for the underlying GP formulation with divergence-free kernels.
Sec.~\ref{secpowerlaws} describes the design of physics-informed kernels for GPH.
Sec.~\ref{secaccap} quantifies the accuracy of the proposed approach as the $L^2$ norm of its residual,  interprets that residual as an instantaneous measure  of information loss (resulting from the discretization of the continuous dynamics), and presents an information loss interpretation and quantification of turbulence.
Sec.~\ref{numericalexp} presents numerical experiments.
Throughout all these sections, we will use figures/simulations from Sec.~\ref{numericalexp} to illustrate the proposed method and refer to  Sec.~\ref{numericalexp} for their detailed descriptions and to \url{https://www.youtube.com/user/HoumanOwhadi} for corresponding animations.

\section{Set up}\label{secsetup}
Let $\T^d$ be the torus of side length $2\pi$ and dimension $d=2$ or $d=3$.
Consider the forced Navier-Stokes equations on $\T^d$
\begin{equation}\label{eqeodjdoije}
\begin{cases}
\partial_t u + u\nabla u=\nu \Delta u-\nabla p +f &\text{ on } \T^d\\
\diiv u=0&\text{ on } \T^d\,,
\end{cases}
\end{equation}
with smooth zero-mean flow\footnote{$u_0\in C^\infty(\T^d)$ and $\int_{\T^d}u_0(x)\,dx=0$. $f\in C^\infty(\T^d\times [0,\infty))$ and $\int_{\T^d}f(x,t)\,dx=0$ for all $t$.} initial conditions $u(x,0)=u_0(x)$ and
external volumetric force $f(x,t)$.

Introducing the vorticity
\begin{equation}
\omega(x,t):=\curl\, u(x,t)\,,
\end{equation}
and $g(x,t)=\curl\, f(x,t)$,

 \eqref{eqeodjdoije} is equivalent to the  equations
 \begin{align}
{\bf (d=2)}\quad &\partial_t \omega + u\nabla \omega=\nu \Delta \omega + g(x,t)\,,\label{eqvor2d}\\
{\bf (d=3)}\quad &\partial_t \omega + u\nabla \omega=\nu \Delta \omega+\omega \nabla u + g(x,t)\label{eqvor3d}\,,
\end{align}
with initial condition $\omega(x,0)=\omega_0(x):=\curl\, u_0(x)$.

\begin{figure}[h]
    \centering
        \includegraphics[width=\textwidth ]{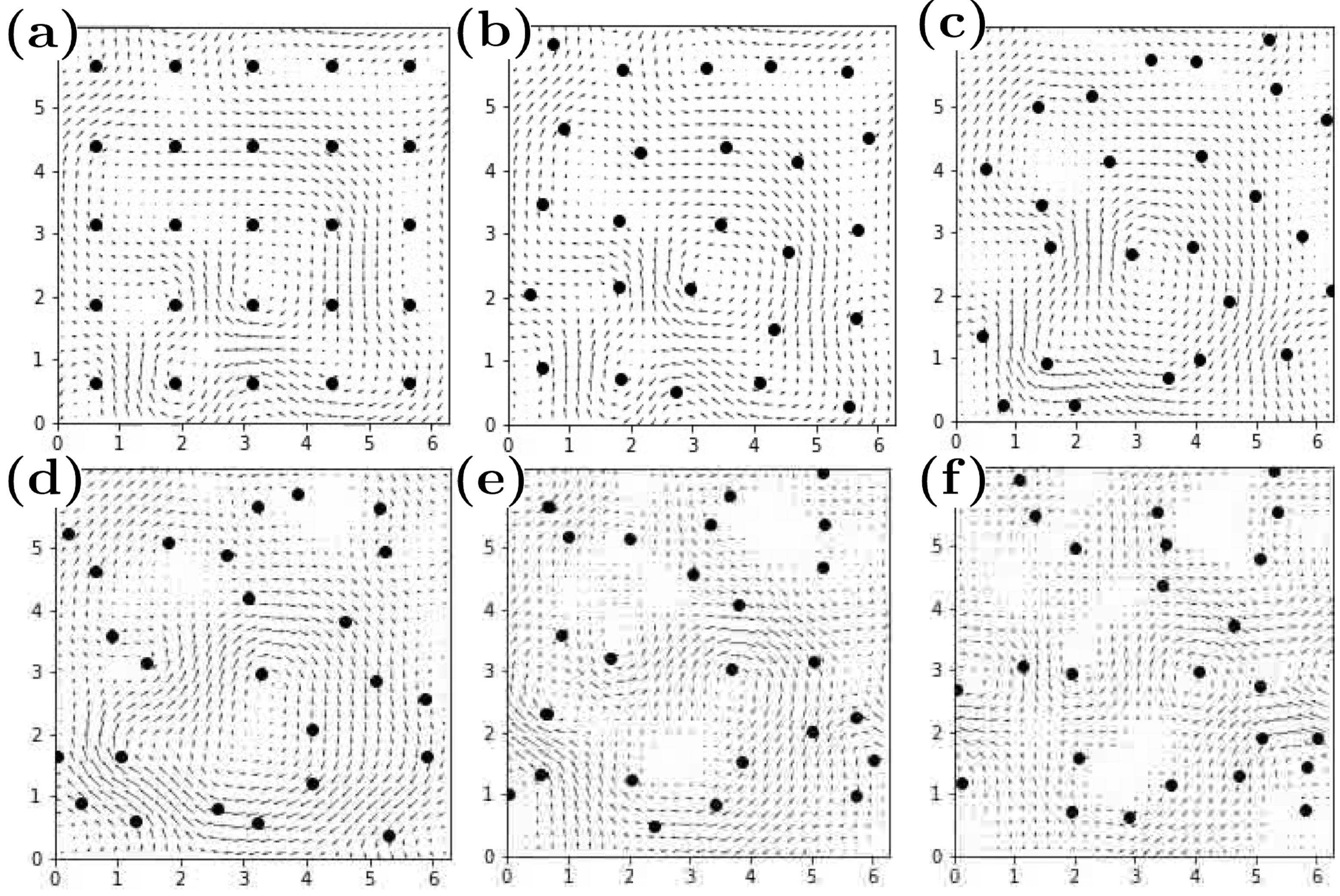}
    \caption{Velocity snapshots at times $t=i*0.3$ for $i=0,1,\ldots,5$.}
    \label{fiex2}
\end{figure}

\begin{figure}[h]
    \centering
        \includegraphics[width=\textwidth ]{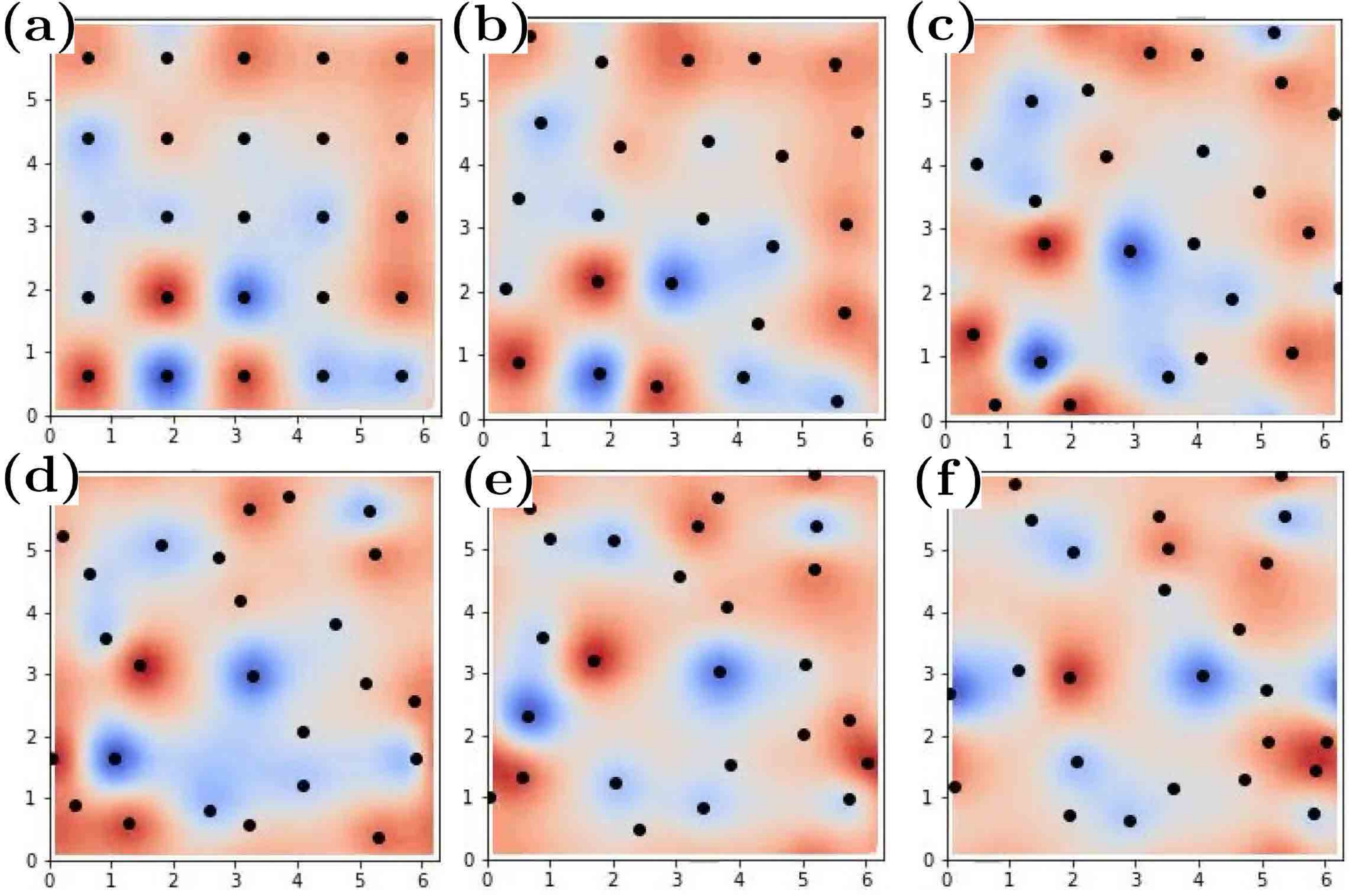}
    \caption{Vorticity snapshots at times $t=i*0.3$ for $i=0,1,\ldots,5$.}
    \label{fiex3}
\end{figure}

\begin{figure}[h]
    \centering
        \includegraphics[width=\textwidth ]{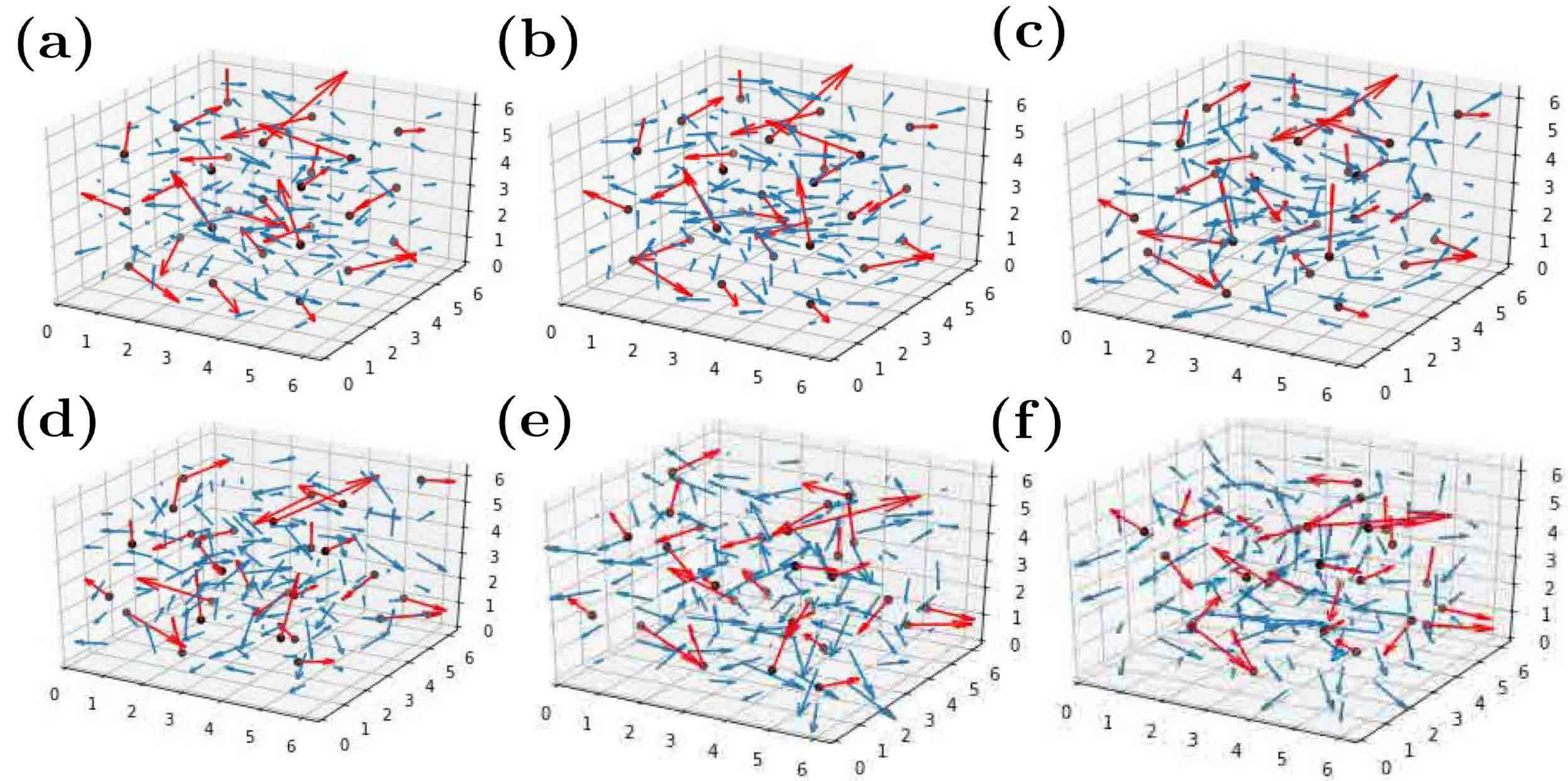}
    \caption{{\bf $d=3$}. Velocity and vorticity snapshots at times $t=i*0.3$ for $i=0,1,\ldots,5$. Blue arrows show velocity, red arrows show vorticity $W$ at particle locations $q$.}
    \label{fiex3d}
\end{figure}

\begin{figure}[h]
    \centering
        \includegraphics[width=\textwidth ]{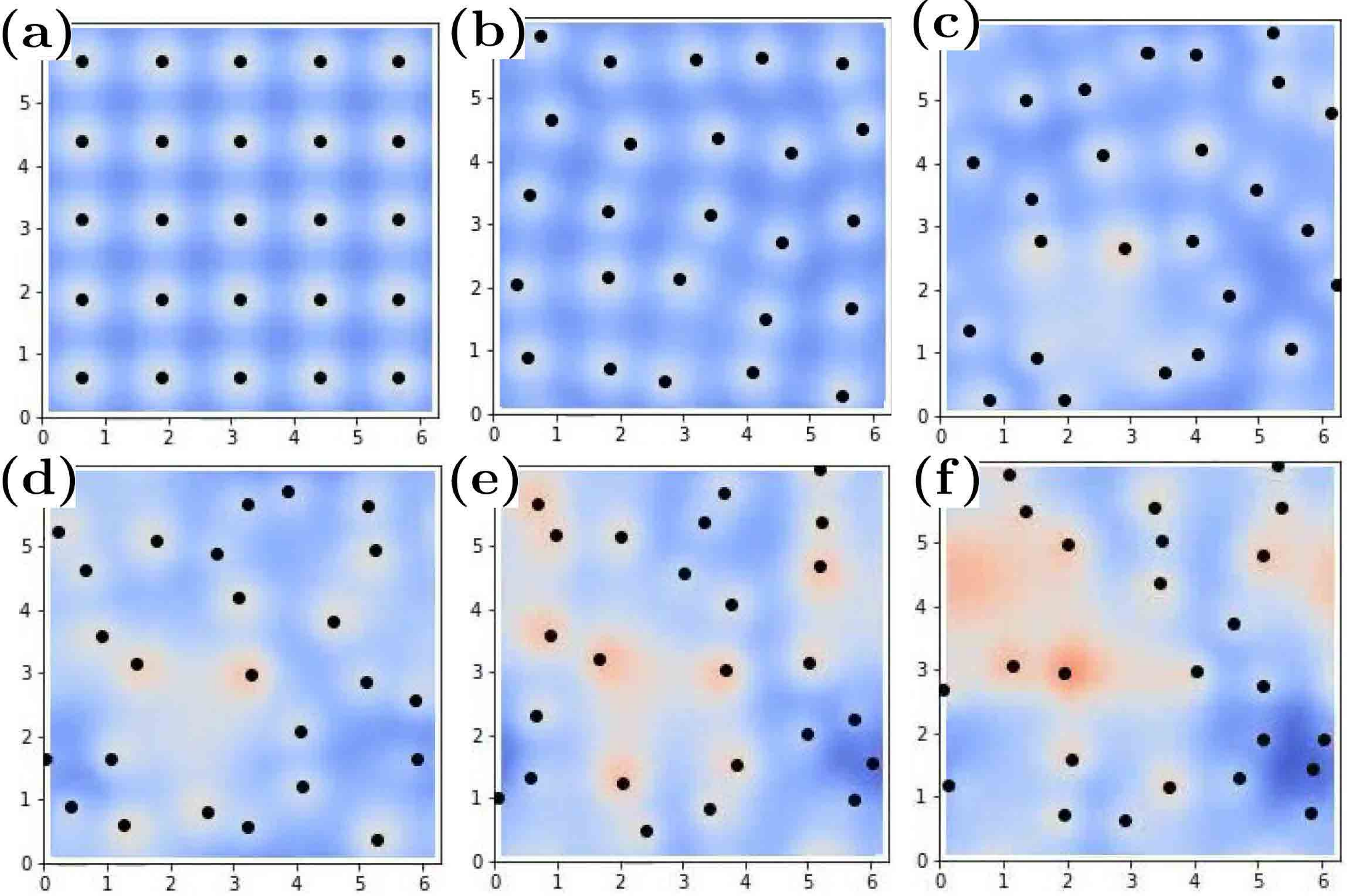}
    \caption{Variance velocity at times $t=i*0.3$ for $i=0,1,\ldots,5$. The color scale is $4.8$ (blue) to $8.4$ (red).}
    \label{fiex4}
\end{figure}

\begin{figure}[h]
    \centering
        \includegraphics[width=\textwidth ]{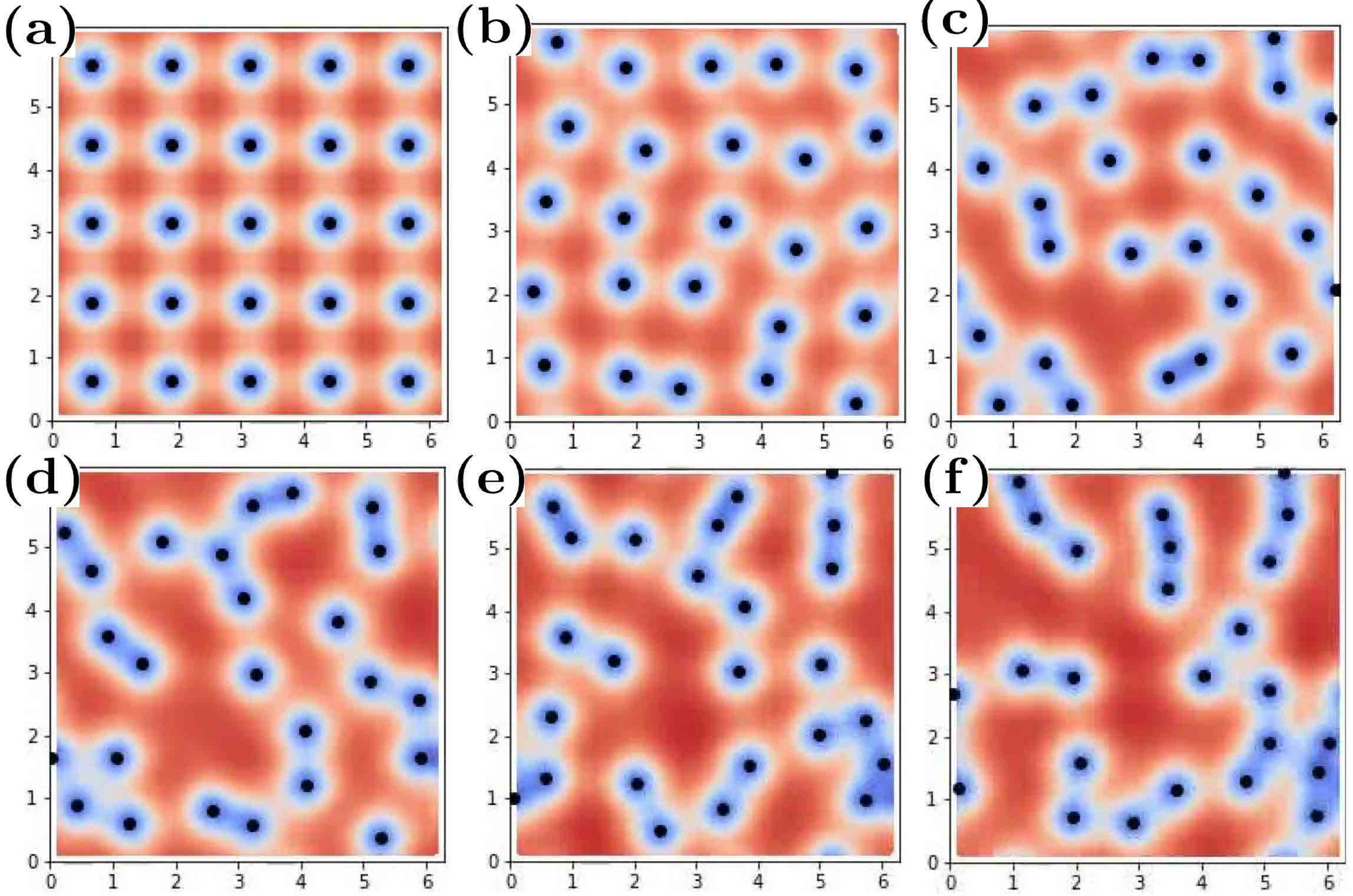}
    \caption{Variance vorticity at times $t=i*0.3$ for $i=0,1,\ldots,5$. The color scale is $0$ (blue) to $45$ (red).}
    \label{fiex5}
\end{figure}

\section{Gaussian Process Hydrodynamics}\label{secGPH}
Let $X_1,\ldots,X_N$ be $N$ distinct (and possibly homogeneously distributed) collocation points $X_i$ in $\T^d$.
For $i\in \{1,\ldots,N\}$ let $t\rightarrow q_i(t)$ be  the trajectory formed by a particle advected by the flow velocity $u(x,t)$, defined as the solution of
\begin{equation}\label{eqjkhhed}
\dot{q}_i(t)=u\big(q_i(t),t\big)
\end{equation}
with initial condition $q^0_i=X_i \in \T^d$.
For $i\in \{1,\ldots,N\}$ let
\begin{equation}\label{eqjhedjhddgy}
W_i(t):=\omega(q_i(t),t)
\end{equation}
be the value of the vorticity at $(q_i(t),t)$.  \eqref{eqjkhhed}, \eqref{eqvor2d} and \eqref{eqvor3d} imply that $t\rightarrow W_i(t)$ solves the ODE
 \begin{align}
{\bf (d=2)}\quad &\dot{W}_i(t)=\nu \Delta \omega(q_i(t),t)+g(q_i(t),t) \,,\label{eqW2d}\\
{\bf (d=3)}\quad &\dot{W}_i(t)=\nu \Delta \omega(q_i(t),t)+W_i(t) \nabla u(q_i(t),t)+g(q_i(t),t) \label{eqW3d}\,,
\end{align}
with initial condition $W_i(0)=\omega_0(q_i(0))$.
Write $q(t):=(q_1(t),\ldots,q_N(t))$ and $W(t):=(W_1(t),\ldots,W_N(t))$.
Since $(q,W)$ only provide partial information on $u$ and its partial derivatives, \eqref{eqW2d} and \eqref{eqW3d} are not autonomous systems and closing them requires closing the information gap between $(q,W)$ and $u$, i.e. approximating $u(x,t)$ and its partial derivatives as a function of $(q,W)$. Our approach to this closure problem is to replace the unknown velocity field $u$ by a centered Gaussian Process (GP) $\xi\sim \cN(0,K)$ (with a physics-informed matrix-valued kernel $K$ that may be non-stationary to incorporate non-periodic boundary conditions) and approximate $u$  with the conditional expectation of $\xi$ given the information \eqref{eqjhedjhddgy}. To describe this, let
 \begin{align}
{\bf (d=2)}\quad &\Y:=(\T^2)^N \times \R^N \,,\label{eqY2d}\\
{\bf (d=3)}\quad &\Y:=(\T^3)^N \times (\R^3)^N \label{eqY3d}\,,
\end{align}
be the phase space containing the trajectory $t\rightarrow (q,W)(t)$. Define
\begin{equation}\label{eqkhgdkejgedejhge}
u^\star\big(x,q,W\big):=\E\big[\xi(x)\big| \curl\, \xi(q)=W\big]\text{ for } (x,q,W)\in \T^d \times \Y\,
\end{equation}
where, using vectorized notations, we have written $\curl\, \xi(q)$ for the $N$-vector with entries $\curl\, \xi(q_i)$.
We then  approximate $(q,W)(t)$ with $(q^\star,W^\star)(t)$,  $u(x,t)$ with
\begin{equation}
\bar{u}(x,t):=u^\star(x,q^\star(t),W^\star(t))\,,
\end{equation}
and  $\omega(x,t)$ with
\begin{equation}
\bar{\omega}(x,t):=\curl\, u^\star(x,q^\star(t),W^\star(t))\,,
\end{equation}
where $(q^\star,W^\star)$ is the solution of the autonomous system of ODEs\footnote{The differential operators $\Delta \curl$ and $\nabla$ in \eqref{eqodesys} act on the first argument $x$ of $u^\star$ in \eqref{eqkhgdkejgedejhge}.}
 \begin{equation}\label{eqodesys}
  \begin{cases}
&\dot{q}^\star_i=u^\star\big(q_i^\star,q^\star,W^\star\big)\\
{\bf (d=2)}\quad &\dot{W}_i^\star(t)=\nu \Delta \curl\, u^\star\big(q_i^\star,q^\star,W^\star\big)+g(q_i^\star(t),t) \,,\\
{\bf (d=3)}\quad &\dot{W}_i^\star(t)=\nu \Delta \curl\, u^\star\big(q_i^\star,q^\star,W^\star\big)+W_i^\star \nabla u^\star\big(q_i^\star,q^\star,W^\star\big)+g(q_i^\star(t),t)\,,
\end{cases}
\end{equation}
with the initial condition $(q^\star,W^\star)(0)=(q,W)(0)=(q^0,\omega_0(q^0))$.
See Fig.~\ref{fiex2} and \ref{fiex3} for snapshots of $\bar{u}$ (shown as a vector field), $\bar{\omega}$ (shown as a heatmap) and $\bar{q}$ (shown as dark points). For $d=3$, see Fig.~\ref{fiex3d} for snapshots of $\bar{u}$ (shown as blue arrows), $W$ (shown as red arrows) and $\bar{q}$ (shown as dark points).

The proposed approach comes with Uncertainty Quantification (UQ) estimates and is compatible with a UQ pipeline. In particular, given $\curl\,\xi(q)=W$,
 $\xi$ is a GP with conditional mean $u^\star$ and conditional covariance function
\begin{equation}\label{eqcovcu}
\C^u(x,y):=\E\Big[\big(\xi(x)-u^\star(x,q,W)\big)\big(\xi(y)-u^\star(x,q,W)\big)^T\Big|\curl\,\xi(q)=W \Big]\,,
\end{equation}
and $\curl\,\xi$ is a GP with conditional mean $\curl\,u^\star$ and conditional covariance function
{
\begin{equation}\label{eqcovcomega}
\C^\omega(x,y):=\E\Big[\big(\curl\,\xi(x)-\curl\,u^\star(x,q,W)\big)\big(\curl\,\xi(y)-\curl\,u^\star(x,q,W)\big)^T\Big|\curl\,\xi(q)=W \Big]\,.
\end{equation}
}
See Fig.~\ref{fiex4} and \ref{fiex5} for snapshots of $x\rightarrow \Tr[\C^u(x,x)]$ and $x\rightarrow \Tr[\C^\omega(x,x)]$ shown as heatmaps.

\section{Divergence free GPs/kernels and representer formulas}\label{secdivfreegp}
We will now describe the vector-valued GP $\xi\sim \cN(0,K)$ employed to close the NS equations and introduce representer formulas for identifying $u^\star$ and its partial derivatives as a function of $(q^\star,W^\star)$.
Recall (see \cite[Chap.~7,17]{owhadi2019operator} and \cite[Sec.~8.1]{owhadi2020ideas}) that $x\rightarrow \xi(x)$ is a map from $\T^d$ to a linear (Hilbert) space of $d$-dimensional centered Gaussian vectors such that
\begin{equation}
\Cov\big(\xi(x),\xi(y)\big)=K(x,y) \text{ for }x,y\in \T^d\,,
\end{equation}
where the covariance function
$K$ is a $\R^{d\times d}$ matrix-valued kernel (also known as a vector-valued kernel \cite{alvarez2012kernels}).
Write $\H_K$ for the reproducing kernel Hilbert space (RKHS) of $\R^d$ valued functions defined by $K$.
To ensure that  our approximation $u^\star$ remains zero-mean, incompressible and that  \eqref{eqkhgdkejgedejhge} and \eqref{eqodesys} are properly defined, we select $K$ so that $\H_K$ is contained in the set
\begin{equation}
\S^3(\T^d):=\big\{v\in C^3(\T^d)\mid \int_{\T^d}v(x)\,dx=0 \text{ and } \diiv v=0 \big\}
\end{equation}
of $\T^d$-periodic zero-mean divergence-free $\R^d$-valued functions with continuous third-order derivatives (we write $C^k$ for the space of continuously $k$th-order differentiable functions). Matrix-valued kernels inducing an RKHS containing divergence-free vector-valued functions can be
 constructed by starting with a stationary scalar-valued kernel $G(x,x')=g(x-x')$ and selecting
$K(x,y)= \big(\Hess g -\Tr[\Hess g] I_d\big)(x-y)$ where $\Hess$ is the Hessian operator and $I_d$ is the $d\times d$ identity matrix \cite[Sec~5.1]{alvarez2012kernels}.
Here we present a more general approach enabling using non-stationary kernels and the incorporation of nontrivial fluid-structure boundary conditions into the kernel (see Fig.~\ref{figobstacle}).
We will distinguish the $d=2$ and $d=3$ cases in our description of this approach.

\begin{figure}[h]
    \centering
        \includegraphics[width=\textwidth ]{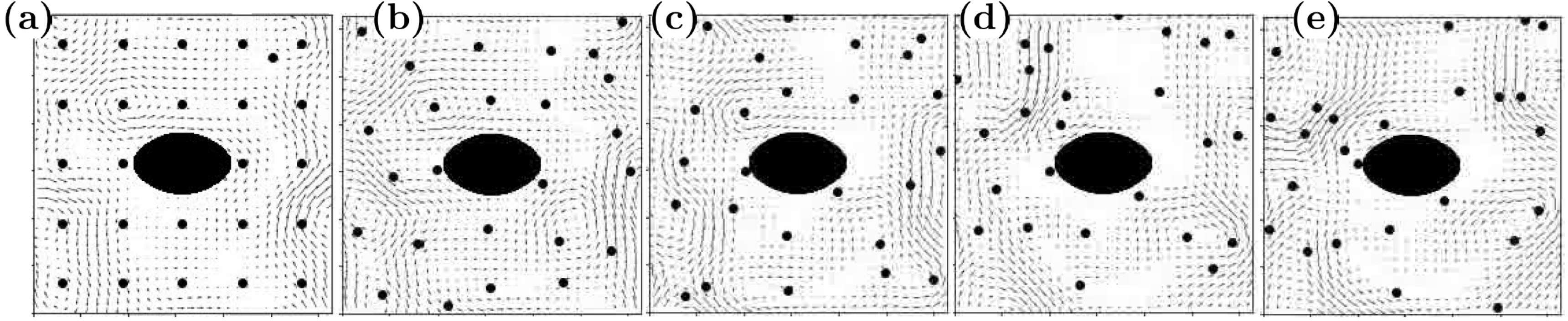}
    \caption{Flow around an obstacle at times $t=i*0.3$ for $i=0,1,\ldots,5$.}
    \label{figobstacle}
\end{figure}

\subsection{Two dimensional case ($d=2$)}\label{subdeq2ker}

\subsubsection{Divergence-free kernels}
Given an $\R^2$-valued  function $v(x)=(v_1(x),v_2(x))^T$, $\curl\, v=(-\partial_{x_2} v_1(x) +\partial_{x_1} v_2(x))$ can be written as the inner product between
the row vector $\curl_x=(-\partial_{x_2}, \partial_{x_1})$ and the column vector $v(x)=(v_1(x),v_2(x))^T$. Let $G$ be a non-degenerate $C^3$-differentiable scalar-valued kernel on $\T^2$ such that $\H_G$ (the RKHS defined by $G$) is compactly embedded in $H^s(\T^2)$ for some $s>5$.
Extending matrix-vector operations to differential operators, we define
\begin{equation}\label{eqkheddy}
K(x,y):=\curl_x^T \curl_y G(x,y)=\begin{pmatrix}-\partial_{x_2}\\ \partial_{x_1}\end{pmatrix}
\begin{pmatrix}-\partial_{y_2}& \partial_{y_1} \end{pmatrix}G(x,y)\,,
\end{equation}
which can also be written as,
\begin{equation}\label{eqK2d}
K(x,y):=\begin{pmatrix}\partial_{x_2}\partial_{y_2}&-\partial_{x_2}\partial_{y_1}\\
-\partial_{x_1}\partial_{y_2}& \partial_{x_1}\partial_{y_1}\ \end{pmatrix}G(x,y):=\begin{pmatrix}\partial_{x_2}\partial_{y_2}G(x,y)&-\partial_{x_2}\partial_{y_1}G(x,y)\\
-\partial_{x_1}\partial_{y_2}G(x,y)& \partial_{x_1}\partial_{y_1} G(x,y) \end{pmatrix}\,.
\end{equation}

The following proposition shows that $K$ is a valid non-degenerate kernel satisfying our requirements.
\begin{Proposition}\label{prop2d}
It holds true that  (1)
$K=$\eqref{eqK2d} is a non-degenerate kernel,  (2) Its RKHS $\H_K$ is compactly embedded in $\H^{s-1}(\T^2)$, and (3)  $\H_K\subset \S^3(\T^2)$.
\end{Proposition}
\begin{proof}
To show that $K$ is a valid kernel, we will employ the one-to-one map between kernels, symmetric positive definite linear operators
, and quadratic norms presented in
\cite[Chap.~11,17]{owhadi2019operator} (see also \cite[Sec.~2.1]{chen2021solving}).
Write  $\H_G$ and $\|\cdot\|_G$ for the RKHS space and the RKHS norm induced by $G$ and $\H_G^*$ and $\|\cdot\|_G^*$ for their duals with respect to the $L^2$ inner product which we write $[\cdot,\cdot]$ ($[\varphi,f]:=\int_{\T^2}\varphi(x)f(x)\,dx$ for $\varphi \in \H_G^*$ and $f\in \H_G$).
The operation $\varphi \rightarrow \int_{\T^2} G(x,y)\varphi(y)\,dy$ defines a linear bijection $\G$ mapping $\H_G^*$ to $\H_G$ that is symmetric ($[\varphi, \G \varphi']=[\varphi', \G\varphi]$), positive ($[\varphi, \G \varphi]\geq 0$) and definite ($[\varphi, \G \varphi]=0$ if and only if $\varphi=0$).
Writing $\updelta_x$ for a delta Dirac function supported at the point $x$, $\G$ defines the kernel $G$ via $G(x,y)=[\updelta_x, \G \updelta_y]$.
Furthermore
 $\|\varphi\|_G^{*,2}=\int_{(\T^2)^2} G(x,y)\varphi(x)\varphi(y)\,dx\,dy=[\varphi,\G \varphi]$ for $\varphi \in \H_G$ and $\|f\|_G=\sup_{\varphi \in \H_G^*}[\varphi,f]/\|\varphi\|_{G}^*$ for $f\in \H_G$. These identities show that there is a one-to-one correspondence between the (non-degenerate) kernel $G$, the symmetric positive definite linear bijection $\G$ and the quadratic norms $\|\cdot\|_G^*$ and $\|\cdot\|_G$ (any of these objects can be used to define a valid kernel \cite[Chap.~11,17]{owhadi2019operator}).
For $\phi \in \S^3(\T^2)$ write
\begin{equation}\label{eqklwhjldkedh}
\|\phi\|_K^{*}:=\|\curl\, \phi\|_G^*\,.
\end{equation}
 Since $\|\phi\|_K^{*}$ is a quadratic norm on $\S^3(\T^2)$ it defines \cite[Chap.~11,17]{owhadi2019operator} a non-degenerate kernel $K$ with RKHS space $\H_K$ and  norm
 $\|\cdot\|_K$ such that $\|\cdot\|_K^{*}$ is the dual of $\|\cdot\|_K$ with respect to the $L^2$ inner product and $\H_K^*$ is the closure of
  $\S^3(\T^2)$ with respect to $\|\cdot\|_K^{*}$ (note that the construction \eqref{eqkheddy} and the identity $\diiv \curl = {\bf 0}$ imply that the elements of $\H_K$ are divergence-free functions).
 For the sake of clarity, we will also present the following alternate proof of the non-degeneracy of $K$.
For  $q\in (\T^2)^N$ write $K(q,q)$ for the $N\times N$ block matrix with $2\times 2$ block entries $K(q_i,q_j)$. For $\alpha\in (\R^2)^N$, write
$\alpha^T K(q,q)\alpha:=\sum_{i,j=1}^N \alpha_i^T K(q_i,q_j)\alpha_j$. Then the identity $\alpha^T K(q,q)\alpha=\big\|\sum_{i=1}^N \updelta_{q_i}\circ (-\alpha_1\partial_{x_2}+\alpha_2 \partial_{x_1})\big\|_G^{*,2}$ implies that $K(q,q)$ is invertible if the $q_i$ are pairwise distinct and $\alpha\not=0$, i.e. $K$ is non-degenerate.
 (2) follows from the identity \eqref{eqklwhjldkedh} and $\|v\|_K=\sup_{\phi \in \H_K^*}[\phi,v]/\|\phi\|_K^*$. (3) follows from (a) the compact embedding of $\H_G$  into $H^s(\T^2)$ for some $s>5$, and (b) the compact embedding of $H^{s-1}(\T^2)$ into $C^3(\T^2)$ for $s>5$.
\end{proof}

\begin{Remark}
The results of this section naturally generalize to the situation where $G$ is a kernel on  $A/B$ where $B$ is an inclusion in the domain $A$.
In that case, required boundary conditions on the elements of $\H_K$ (e.g., stick or no-slip) transfer onto required boundary conditions on $G$.
Possible designs of $G$ include (1) identifying $G$ as the Green's function of a higher order elliptic PDE on $A/B$ with the required boundary conditions, (2) designing $G$ with transformations of an initial defined on $\R^d$. For Fig.~\ref{figobstacle}, $G(x,y)=\g(x,y) f(x)f(y)$ where $\g$ is a kernel on $\T^2$ and $f$ is a smooth function equal to zero on the inclusion $B$ and one on $\T^2/B_\epsilon$ where $B_\epsilon$ is an $\epsilon$ enlargement of $B$ obtained by adding a boundary layer of size $\epsilon$ (the resulting elements of $\H_K$ satisfy a stick boundary condition).
\end{Remark}

\subsubsection{Representer formulas}\label{subsecrep2d}
We will now introduce representer formulas for the conditional mean and covariance of the GP $\xi \sim \cN(0,K)$ given $\curl\, \xi(q)=W$.
Write $\Delta_x \Delta_y G$ for the  $\T^2$ valued kernel $\Delta_x \Delta_y G(x,y)$.
For $q\in (\T^2)^N$ write $\Delta_x \Delta_y G(q,q)$ for the $N\times N$ matrix with entries $\Delta_x \Delta_y G(q_i,q_j)$.
For $x\in \T^2$ and $q\in \T^N$ write $\curl_x^T \Delta_y G(x,q)$ for the $N$-vector with $\R^2$-valued entries
$\curl_x^T \Delta_y G(x,q_i)$.

\begin{Proposition}\label{propkefhirfh}
The GP $\xi \sim \cN(0,K)$ conditioned on $\curl\, \xi(q)=W$ is also Gaussian with conditional mean
\begin{equation}\label{eqjkhkwhb1}
u^\star(x,q,W)=\eqref{eqkhgdkejgedejhge}=\curl_x^T \Delta_y G(x,q) \big(\Delta_x \Delta_y G(q,q)\big)^{-1} W\,,
\end{equation}
and conditional covariance kernel $\C^u(x,y)=$\eqref{eqcovcu} given by
\begin{equation}\label{eqjkhkwhb2}
\curl_x^T \curl_y G(x,y)-\curl_x^T \Delta_y G(x,q) \big(\Delta_x \Delta_y G(q,q)\big)^{-1} \Delta_x \curl_y G(q,y) \,.
\end{equation}
Furthermore, the GP $\curl\,\xi $ conditioned on $\curl\, \xi(q)=W$ is also Gaussian with conditional mean
\begin{equation}\label{eqjkhkwhb1curl}
\curl u^\star(x,q,W)=\Delta_x \Delta_y G(x,q) \big(\Delta_x \Delta_y G(q,q)\big)^{-1} W\,,
\end{equation}
and conditional covariance kernel $\C^\omega(x,y)=$\eqref{eqcovcomega} given by
\begin{equation}\label{eqjkhkwhb2curl}
\Delta_x \Delta_y G(x,y)-\Delta_x \Delta_y G(x,q) \big(\Delta_x \Delta_y G(q,q)\big)^{-1} \Delta_x \Delta_y G(q,y) \,.
\end{equation}
\end{Proposition}
\begin{proof}
\eqref{eqjkhkwhb1} and \eqref{eqjkhkwhb2} follow from the
generalized representer theorem  \cite[Cor.~17.12]{owhadi2019operator} (see also \cite[Prop.~2.1]{chen2021solving}) and the identity
 $\curl \curl^T =\Delta $.
For $\alpha \in \R^N$, $\alpha^T \Delta_x \Delta_y G(q,q)\alpha=\|\sum_{i=1}^N \alpha_i \updelta_{q_i}\circ \Delta \|_{G}^{*,2}$ implies that $\Delta_x \Delta_y G(q,q)$ is invertible if the $q_i$ are pairwise distinct.
 \eqref{eqjkhkwhb1curl} and the identities $\curl \curl^T =\Delta $ and
$\C^\omega(x,y)=\curl_x \curl_y^T\C^u(x,y)$  imply \eqref{eqjkhkwhb1curl} and \eqref{eqjkhkwhb2curl}.
\end{proof}
Using Prop.~\ref{propkefhirfh},
\eqref{eqodesys} reduces to
 \begin{equation}\label{eqodesys2d}
  \begin{cases}
&\dot{q}^\star_i=\curl_x^T \Delta_y G(q^\star_i,q^\star) \big(\Delta_x \Delta_y G(q^\star,q^\star)\big)^{-1} W^\star\\
&\dot{W}_i^\star=\nu \Delta_x^2 \Delta_y G(q^\star_i,q^\star) \big(\Delta_x \Delta_y G(q^\star,q^\star)\big)^{-1} W^\star +g(q_i^\star(t),t)\,.
\end{cases}
\end{equation}

\subsection{Three dimensional case ($d=3$)}\label{subdeq2kerd3}

\subsubsection{Divergence-free kernels}
Given an $\R^3$-valued  function $v(x)$, $\curl\, v$ can be written as the inner product between
the matrix
\begin{equation}
\curl_x=\begin{pmatrix} 0 & -\partial_{x_3} & \partial_{x_2} \\ \partial_{x_3} &0 & -\partial_{x_1}\\ -\partial_{x_2}& \partial_{x_1} &0 \end{pmatrix}
\end{equation}
and the column vector $v(x)=(v_1(x),v_2(x), v_3(x))^T$. Let $G$ be a non-degenerate $C^3$-differentiable scalar-valued kernel on $\T^3$ such that $\H_G$ (the RKHS defined by $G$) is compactly embedded in $H^s(\T^3)$ for some $s>5.5$.
Define
\[
K(x,y):=\curl_x^T \curl_y G(x,y)=\begin{pmatrix} 0 & \partial_{x_3} & -\partial_{x_2} \\ -\partial_{x_3} &0 & \partial_{x_1}\\ \partial_{x_2}& -\partial_{x_1} &0 \end{pmatrix}
\begin{pmatrix} 0 & -\partial_{y_3} & \partial_{y_2} \\ \partial_{y_3} &0 & -\partial_{y_1}\\ -\partial_{y_2}& \partial_{y_1} &0 \end{pmatrix}G(x,y)\,,
\]
which can also be written as,
\begin{equation}\label{eqK3d}
K(x,y):=\begin{pmatrix}\partial_{x_3}\partial_{y_3}+\partial_{x_2}\partial_{y_2}&-\partial_{x_2}\partial_{y_1}&-\partial_{x_3}\partial_{y_1}\\
-\partial_{x_1}\partial_{y_2}& \partial_{x_3}\partial_{y_3}+\partial_{x_1}\partial_{y_1} &  -\partial_{x_3}\partial_{y_2}\\
-\partial_{x_1}\partial_{y_3}& -\partial_{x_2}\partial_{y_3}& \partial_{x_1}\partial_{y_2}+\partial_{x_2}\partial_{y_2}
\end{pmatrix}G(x,y)\,.
\end{equation}

\begin{Proposition}\label{prop3d}
It holds true that (1) $K=$\eqref{eqK3d} is a non-degenerate kernel,  (2) Its RKHS $\H_K$ is compactly embedded in $\H^{s-1}(\T^3)$, and (3)  $\H_K\subset \S^3(\T^3)$.
\end{Proposition}

\begin{proof}
Write $\Hess$ for the Hessian operator and $I_3$ for the $3\times 3$ identity matrix.
Integrating by parts,
observe that for $\phi \in \S^3(\T^3)$,
\begin{equation}
\|\phi\|_{K}^{*,2}=\|-\partial_{x_3} \phi_2+\partial_{x_2}\phi_1\|_G^{*,2}+\|\partial_{x_3} \phi_1-\partial_{x_1}\phi_3\|_G^{*,2}
+\|-\partial_{x_2} \phi_1+\partial_{x_1}\phi_2\|_G^{*,2}\,.
\end{equation}
The remaining part of the proof is identical to that of Prop.~\ref{prop2d}.
\end{proof}

\subsubsection{Representer formulas}
We will now present representer formulas for the conditional mean and covariance of the GP $\xi \sim \cN(0,K)$ given $\curl\, \xi(q)=W$.
 Write $\L_x:=I_3 \Delta_x-\Hess_x$ and $\L_x \L_y G$ for the $3\times 3$ matrix valued kernel obtained by letting $\L_x$ act on the $x$ variable and $\L_y$ act on the $y$ variable of $G(x,y)$.
Similarly write $\curl_x^T \L_y G$ for the $3\times 3$ matrix valued function of $x,y$ obtained by letting $\curl_x^T$ act on the $x$ variable and
$\L_y$ on the $y$ variable of $G(x,y)$. Using the shorthand notations of Sec.~\ref{subsecrep2d}, for $q\in (\T^3)^N$, we write $\L_x \L_y G(q,q)$ for the $N\times N$ block matrix whose entries are the $3\times 3$ matrices $\L_x \L_y G(q_i,q_j)$. Similarly we write  $\curl_x^T \L_y G(x,q)$ for the
$N$-block vector whose entries are the $3\times 3$ matrices $\curl_x^T \L_y G(x,q_i)$.

\begin{Proposition}\label{propedkjhedbdy}
The GP $\xi \sim \cN(0,K)$ conditioned on $\curl\, \xi(q)=W$ is also Gaussian with conditional mean
\begin{equation}\label{eqjkhkwhb1b}
u^\star(x,q,W)=\eqref{eqkhgdkejgedejhge}=\curl_x^T \L_y G(x,q) \big(\L_x \L_y G(q,q)\big)^{-1} W\,,
\end{equation}
and conditional covariance kernel $\C^u(x,y)=$\eqref{eqcovcu} given by
\begin{equation}\label{eqjkhkwhb2b}
\curl_x^T \curl_y G(x,y)-\curl_x^T \L_y G(x,q) \big(\L_x \L_y G(q,q)\big)^{-1} \L_x \curl_y G(q,y) \,.
\end{equation}
Furthermore, the GP $\curl\,\xi $ conditioned on $\curl\, \xi(q)=W$ is also Gaussian with conditional mean
\begin{equation}\label{eqjkhkwhb1bcurl}
\curl u^\star(x,q,W)= \L_x\L_y G(x,q) \big(\L_x \L_y G(q,q)\big)^{-1} W\,,
\end{equation}
and conditional covariance kernel $\C^\omega(x,y)=$\eqref{eqcovcomega} given by
\begin{equation}\label{eqjkhkwhb2bcurl}
\L_x \L_y G(x,y)-\L_x \L_y G(x,q) \big(\L_x \L_y G(q,q)\big)^{-1} \L_x \L_y G(q,y)  \,.
\end{equation}
\end{Proposition}
\begin{proof}
\eqref{eqjkhkwhb1b} and \eqref{eqjkhkwhb2b} follow from the
generalized representer theorem  \cite[Cor.~17.12]{owhadi2019operator} (see also \cite[Prop.~2.1]{chen2021solving}) and the identity
 $\curl\, \curl^T =I_3 \Delta -\Hess $. Write $\delta_{k,m}$ for the Kronecker delta ($=1$ for $k=m$ and $=0$ otherwise).
For $\alpha \in (\R^3)^N$, the identity  \[\alpha^T \L_x \L_y G(q,q)\alpha=\sum_{m=1}^3 \Big\|\sum_{i=1}^N \sum_{k=1}^3  \updelta_{q_i}\circ \big[\alpha_{i,k}(\delta_{k,m}\Delta_x-\partial_{x_k}\partial_{x_m})\big]\Big\|_{G}^{*,2}\,,\] implies that $\L_x\L_y G(q,q)$ is invertible if the $q_i$ are pairwise distinct.
\eqref{eqjkhkwhb1b} and the identities $\curl \curl^T =\L $ and $\C^\omega(x,y)=\curl_x \curl_y^T\C^u(x,y)$ imply \eqref{eqjkhkwhb1bcurl} and
\eqref{eqjkhkwhb2bcurl}.
\end{proof}

Using Prop.~\ref{propedkjhedbdy},
\eqref{eqodesys} reduces to
 \begin{equation}\label{eqodesys2dt}
  \begin{cases}
\dot{q}^\star_i=&\curl_x^T \L_y G(q_i^\star,q^\star) \big(\L_x \L_y G(q^\star,q^\star)\big)^{-1} W^\star\\
\dot{W}_i^\star=&\nu \Delta_x \L_x \L_y G(q_i^\star,q^\star) \big(\L_x \L_y G(q^\star,q^\star)\big)^{-1} W^\star \\&+ W^\star_i
\nabla_x \curl_x^T \L_y G(q_i^\star,q^\star) \big(\L_x \L_y G(q^\star,q^\star)\big)^{-1} W^\star+g(q_i^\star(t),t)\,.
\end{cases}
\end{equation}

\subsection{Periodic kernels}\label{secperker}
We will now describe the construction of the kernel $G$, which must be a non-generate $C^3$-differentiable scalar-valued kernel on $\T^d$ such that
 $\H_G$ is compactly embedded in $H^s(\T^d)$ for some $s>4+d/2$.
One approach to designing $G$ is to compose a (sufficiently regular and non-degenerate) kernel $\g$ on $\R^{2d}\times \R^{2d}$ with the function
$h\,:\, \T^d \rightarrow \R^{2d}$ defined by
\begin{equation}
h(x)=\big(\cos(x_1), \sin(x_1),\ldots, \cos(x_d),\sin(x_d)\big)\,,
\end{equation}
and obtain
\begin{equation}
G(x,y)=\g\big(h(x),h(y)\big)\,.
\end{equation}
Taking $\g$ to be the Gaussian kernel
 $\g(X,Y)=\exp(-\frac{|X-Y|^2}{2\sigma^2})$ leads to
\begin{equation}\label{eqljkekjdns}
G(x,y)=\exp\Big(\frac{-d+\sum_{i=1}^d \cos(x_i-y_i) }{\sigma^2}\Big)
\end{equation}
which satisfies the requirements on $G$.

\begin{Remark}
Assume $\g$ to be analytic. It follows \cite{sun2008reproducing} that the elements of its RKHS $\H_\g$ are analytic functions.
Therefore, for every function of the form $f\circ h$ with $f\in \H_\g$ and $\g$ analytic, $f$ is uniquely determined by its values on the range of $h$.
Write $\<\cdot,\cdot\>_\g$ ($\|\cdot\|_\g$) for the RKHS inner product (norm) defined by $\g$.
For $f\in \H_\g$ we can therefore define the norm $\|f\circ h\|:=\|f\|_\g^2$ and write $\<\cdot,\cdot\>$ for its associated inner product.
The  reproducing property
\begin{equation}\label{eqkejhhdekh}
\<f\circ h,\g(h(\cdot),h(x))\>=\<f,\g(\cdot,h(x))\>_\g=f\circ h(x)
\end{equation}
   for  $f\in \H_\g$, implies that $\H_G=\{f\circ h\mid f\in \H_\g\}$ and the RKHS norm $\|\cdot\|_G$ defined by $G$ is $\|\cdot\|$, i.e.,
\begin{equation}
\|f\circ h\|_G^2=\|f\|_\g^2\text{ for }f\in \H_\g\,.
\end{equation}
If $\g$ is not analytic then these results generalize to $\H_G=\{f\circ h\mid f\in \H_\g\}$ with
\begin{equation}\label{eqedlekddlkj}
\|v\|_G^2=\inf_{f\in\H_\g\,:\, f\circ h=v}\|f\|_\g^2\,.
\end{equation}
To show this, observe that since $\{f\in\H_\g\,:\, f\circ h=v\}$ is a closed affine subspace of $\H_\g$, the infimum in \eqref{eqedlekddlkj} is achieved and can be expressed as $P v$ where $P$ is a linear operator. Therefore
 $\|v\|^2=\|P v\|_\g^2$ and $\<v,v'\>=\<P v,P v'\>_\g$
define a quadratic norm and an inner product on  $\{f\circ h\mid f\in \H_\g\}$ satisfying the reproducing identity
\begin{equation}\label{eqiuegdueyd3}
\<v,\g(h(\cdot),h(x))\>=\<P v,P \g(h(\cdot),h(x))\>_\g=\<P v,\g(\cdot,h(x))\>_\g=(Pv)\circ h(x)=v(x)\,,
\end{equation}
 which establishes \eqref{eqedlekddlkj}.
The identity  $P \g(h(\cdot),h(x))=\g(\cdot,h(x))$ employed in \eqref{eqiuegdueyd3} follows by observing that the identity
$\|\g(\cdot,h(x))+f\|_\g^2=\g(h(x),h(x))+\|f\|_\g^2+2 f\circ h(x)$ implies that the minimizer of
$\|\g(\cdot,h(x))+f\|_\g^2$ over $f\in\H_\g$ such that $f\circ h=0$ is $f=0$.
\end{Remark}

\section{Power-laws informed kernels}\label{secpowerlaws}
We will now investigate the incorporation of known scaling/power laws  into the selection of the kernel $G$ introduced in
 Sec.~\ref{secdivfreegp} to derive the divergence-free kernel $K$.
 We will focus on the {\bf two-third law}  derived by Kolmogorov \cite{kolmogorov1941local}  from symmetry and universality assumptions on fully developed (homogeneous and isotropic) turbulence.

In {\bf dimension ${\mathbf d=3}$}, the {\bf two-third law} of fully developed (homogeneous and isotropic) turbulence states that the mean of the velocity increment
$|u(x+y,t)-u(x,t)|^2$  behaves approximately as $|y|^\frac{2}{3}$, the two-thirds power of the distance $|y|$ between the points $x+y$ and $x$ \cite[Chap.~5]{frisch1995turbulence}, which ``is equivalent to the statement that the energy spectrum follows a $k^{-\frac{5}{3}}$ law over a suitable range'' \cite[p.~61]{frisch1995turbulence}.

In {\bf dimension ${\mathbf d=2}$}, the statistics of the velocity increments follow a different power-law \cite{lindborg1999can, boffetta2012two}:
the mean of the squared velocity increment $|u(x+y,t)-u(x,t)|^2$
 behaves approximately as $|y|^2$,
 which is equivalent to the statement that the energy spectrum follows a $k^{-3}$ law  \cite[p.~56]{frisch1995turbulence}.

To incorporate these power laws, observe that, in the proposed GP approach, the velocity $u$ is randomized according to the distribution of $\xi\sim \cN(0,K)$. We will therefore use the identity
\begin{equation}\label{eqlwelejdkledje}
\E\big[|\xi(x)-\xi(y)|^2\big]=\Tr\big[K(x,x)+K(y,y)-2 K(x,y)\big]
\end{equation}
to incorporate the velocity increments power laws discussed above.
Considering the situation where $G(x,y)$ is, as in \eqref{eqljkekjdns}, stationary, i.e., $G(x,y)=\psi(x-y)$ for some function $\psi$,
\eqref{eqK2d} and \eqref{eqK3d} reduce to the particular construction of \cite[Sec~5.1]{alvarez2012kernels}, i.e.,
\begin{equation}\label{eqkjhegdedg}
K(x,y)= \big(\Hess \psi -\Tr[\Hess \psi] I_d\big)(x-y)\,.
\end{equation}
\eqref{eqlwelejdkledje} then reduces to
\begin{equation}\label{eqjhgfffrw}
\E\big[|\xi(x)-\xi(y)|^2\big]=2(d-1)\big( \Delta \psi(x-y)-\Delta \psi(0)\big)\,.
\end{equation}

\subsection{The Richardson cascade}\label{secsecrcas}
The basic phenomenology of turbulence, known as the   Richardson cascade \cite[Chap.~7]{frisch1995turbulence}, is that the velocity field is composed of space-filling eddies of various sizes $ \ell_0 r^n$ for some $0<r<1$ and $n=0,1,\ldots,m$. This phenomenology is associated with the concept of energy cascade, representing the idea that the energy is transferred from large (inertial) scales of motion to the small (dissipative) scales. Two-dimensional turbulence is also associated with the possible presence of an inverse energy cascade \cite{sommeria1986experimental} representing the transfer of energy from the small scales to the large scales.
The dissipation scale $\ell_c\sim \ell_0 r^{m}$ is identified by  matching  the  convective transport time scale ($\ell/\delta v(\ell)$) with the diffusive transport time scale ($\ell^2/\nu$). For $d=3$ (using $\delta v(\ell)\sim U_0 (\ell/\ell_0)^\frac{1}{3}$), this translates into $\ell_c/\ell_0\sim R^{-\frac{3}{4}}  $ where $R= U_0 \ell_0 /\nu$ is the Reynolds number. For $d=2$ (using $\delta v(\ell)\sim U_0 (\ell/\ell_0)$), this translates into $\ell_c/\ell_0\sim R^{-\frac{1}{2}}$.
We incorporate these  concepts from a statistical perspective by representing the GP $\xi\sim \cN(0,K)$ as the additive GP
\begin{equation}\label{eqkjedkdjd}
\xi=\sum_{n=0}^m \xi^{(n)}\,,
\end{equation}
where the $\xi^{(n)}\sim \cN(0,K^{(n)})$ are independent and represent eddies at the scale indexed by $n$.
Representing $\xi$ as an additive GP is equivalent to representing $K$ as the additive kernel
\begin{equation}
K=\sum_{n=0}^m K^{(n)}\,.
\end{equation}
To ensure that $K$ and the $K^{(n)}$ are  divergence-free matrix valued kernels we select, as in Sec.~\ref{secdivfreegp},
$K^{(n)}(x,y):=\curl_x^T \curl_y G^{(n)}(x,y)$ where $G^{(n)}$ is a  periodic scalar-valued kernel on $\T^d$.
This is equivalent to selecting  $K(x,y):=\curl_x^T \curl_y G(x,y)$ with
\begin{equation}
G=\sum_{n=0}^m G^{(n)}\,.
\end{equation}

\subsection{Power laws}\label{secpwlaws}
We will now incorporate the velocity-increments power laws into the selection of the kernels $G^{(n)}$.
To incorporate periodicity, stationarity, power laws, and self-similarity we select (as in \eqref{eqljkekjdns})
\begin{equation}\label{eqljkekjdns2}
G^{(n)}(x,y)=\alpha_n \exp\Big(\frac{-d+\sum_{i=1}^d \cos(x_i-y_i) }{\sigma_n^2}\Big)\,,
\end{equation}
with
\begin{equation}\label{eqjhedhhded}
\sigma_n= \frac{\sigma_0}{2^n}\text{ and }\alpha_n=\sigma_n^\gamma \,,
\end{equation}
 for some $\gamma\in \R$  to be determined by the power-law discussed in Sec.~\ref{secpowerlaws}.
For
\begin{equation}
\psi^{(n)}(x):=\alpha_n \exp\Big(\frac{-d+\sum_{i=1}^d \cos(x_i) }{\sigma_n^2}\Big)\,,
\end{equation}
we have
\begin{equation}
\Delta \psi^{(n)}(x):=\sum_{j=1}^d \alpha_n \Big(\frac{\sin^2(x_j) }{\sigma_n^4}-\frac{\cos(x_j) }{\sigma_n^2}\Big)  \exp\Big(\frac{-d+\sum_{i=1}^d \cos(x_i) }{\sigma_n^2}\Big)\,,
\end{equation}
which by \eqref{eqjhgfffrw} leads to
\begin{equation}\label{eqjhgfffrw2}
\begin{split}
\E\big[|\xi(x)-\xi(0)|^2\big]=2(d-1)\sum_{n=0}^m \alpha_n &\sum_{j=1}^d\Bigg(\frac{1 }{\sigma_n^2} +
\Big(\frac{\sin^2(x_j) }{\sigma_n^4}-\frac{\cos(x_j) }{\sigma_n^2}\Big)\\&  \exp\Big(\frac{-d+\sum_{i=1}^d \cos(x_i) }{\sigma_n^2}\Big)
\Bigg)\,.
\end{split}
\end{equation}
We deduce that for  $x\sim 2^{-q}$ with $1<q<m$,
\begin{equation}\label{eqjhgfffrw3}
\begin{split}
\E\big[|\xi(x)-\xi(0)|^2\big]\sim 2(d-1)\sum_{n=q}^m \alpha_n \frac{d }{ \sigma_n^2}
\,.
\end{split}
\end{equation}
Observing that $\sigma_n \sim 2^{-n}$ and $\alpha_n \sim 2^{-n \gamma}$, it follows that for $ |x|\sim 2^{-q}$ and $\gamma>2$,
\begin{equation}\label{eqjhgfffrw4}
\begin{split}
\E\big[|\xi(x)-\xi(0)|^2\big]\sim  2^{q (2-\gamma)}\,.
\end{split}
\end{equation}
Therefore the velocity increment power laws of Sec.~\ref{secpowerlaws} can be incorporated by taking
\begin{equation}
\begin{cases}
\gamma=4 &\text{ for } d=2\\
\gamma=\frac{2}{3}+2&\text{ for } d=3\,.
\end{cases}
\end{equation}

\subsection{Mode decomposition}\label{secma}
Although the Richardson cascade is based on a qualitative analysis of turbulence supported by a qualitative notion of eddies at different scales, this analysis can be made quantitative through kernel mode decomposition \cite{owhadi2019kernelmd}. To describe this, observe that the
 decomposition \eqref{eqkjedkdjd} leads to a corresponding decomposition of the velocity field \eqref{eqkhgdkejgedejhge}, i.e.,
\begin{equation}\label{eqkhgdkejgedejhgeb2}
u^\star\big(x,q,W\big)=\sum_{n=0}^m u^{(n)}\big(x,q,W\big)
\end{equation}
with
\begin{equation}\label{eqkhgdkejgedejhgeb3}
u^{(n)}\big(x,q,W\big)=\E\big[\xi^{(n)}(x)\big| \curl\, \xi(q)=W\big]\,,
\end{equation}
where $u^{(n)}$ the following representer formulas (using the notations of \eqref{eqjkhkwhb1} and \eqref{eqjkhkwhb1b})
\begin{equation}\label{eqjkhkwhb1dd}
u^{(n)}(x,q,W)=
\begin{cases}
\curl_x^T \Delta_y G^{(n)}(x,q) \big(\Delta_x \Delta_y G(q,q)\big)^{-1} W\quad &{\mathbf d=2}\,,\\
\curl_x^T \L_y G^{(n)}(x,q) \big(\L_x \L_y G(q,q)\big)^{-1} W \quad &{\mathbf d=3}\,.
\end{cases}
\end{equation}
Furthermore the RKHS norm of $u^\star$ admits the decomposition  \cite[Sec.~4.1]{owhadi2019kernelmd}
\begin{equation}\label{eqkhgdkejgedejhgeb2x}
\|u^\star\|_K^2=\sum_{n=0}^m \|u^{(n)}\|_{K^{(n)}}^2
\end{equation}
where
\begin{equation}
\|u^{(n)}\|_{K^{(n)}}^2=\<u^{(n)},u^*\>_K=\Var\big[\<\xi^{(n)},u^*\>_K\big]\,,
\end{equation}
can be interpreted as a measure of the activation of the GP (mode) $\xi^{(n)}$ after conditioning on $\xi(q)=W$.
Using $K^{(n)}(x,y):=\curl_x^T \curl_y G^{(n)}(x,y)$, we obtain
\begin{equation}\label{eqjkhkwhb1ddrn}
\big\|u^{(n)}(\cdot,q,W)\big\|_{K^{(n)}}^2=
\begin{cases}
W^T\big(\Delta_x \Delta_y G(q,q)\big)^{-1} W\quad &{\mathbf d=2}\,,\\
W^T \big(\L_x \L_y G(q,q)\big)^{-1} W \quad &{\mathbf d=3}\,.
\end{cases}
\end{equation}

\begin{figure}[h]
    \centering
        \includegraphics[width=\textwidth ]{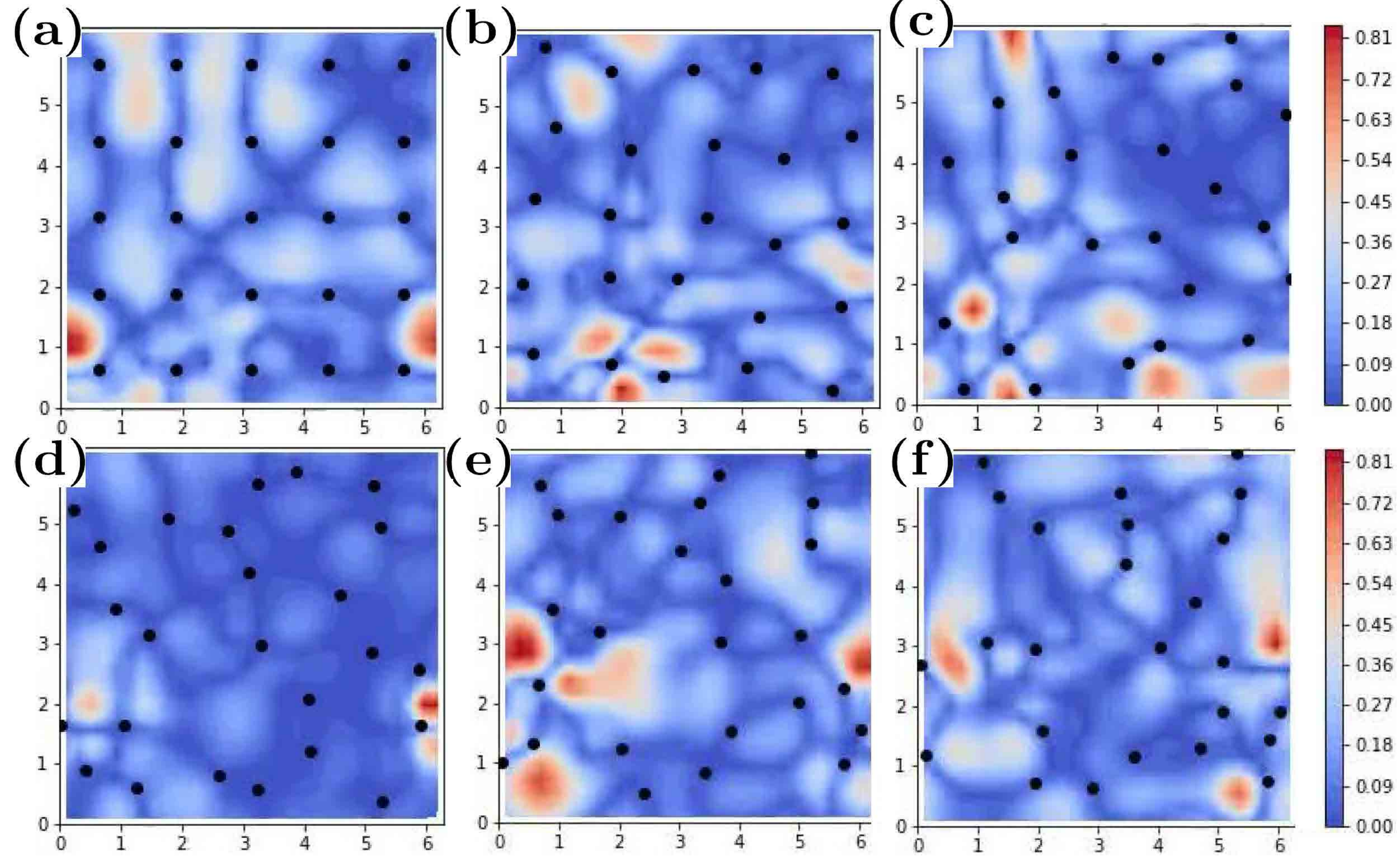}
    \caption{Source terms error $\s(\cdot,t)$ at times $t=i*0.3$ for $i=0,1,\ldots,5$.}
    \label{fiex6}
\end{figure}

\section{Accuracy of the proposed approach and information loss}\label{secaccap}

\subsection{The residual (source term error) as a measure of accuracy}
The accuracy of the proposed approach can be characterized by two terms. The first one is the error $\bar{\omega}(x,0)-\omega(x,0)$ in approximating the initial value of the vorticity.
The second term is the spurious source term $\s$ introduced by the numerical method, defined as (see Fig.~\ref{fiex6} for snapshots\footnote{We are using periodic boundary conditions, so the errors in these snapshots are solely a solely a reflection of the particle locations and the initial condition.} of $|\s(\cdot,t)|$).
 \begin{align}
{\bf (d=2)}\quad &\s(x,t):=\partial_t \bar{\omega} + \bar{u}\nabla \bar{\omega}-\nu \Delta \bar{\omega} - g(x,t)\,,\label{eqvor2ds}\\
{\bf (d=3)}\quad &\s(x,t):=\partial_t \bar{\omega} + \bar{u}\nabla \bar{\omega}-\nu \Delta \bar{\omega}-\bar{\omega} \nabla \bar{u} - g(x,t)\label{eqvor3ds}\,,
\end{align}

The first term $\bar{\omega}(x,0)-\omega(x,0)$ is well-understood as a kernel interpolation error, and a-priori error estimates can be obtained from
 Poincar\'{e} inequalities \cite{Wendland:2005, owhadi2019operator}: the norm of this term can be shown to decay towards zero as a power of the fill distance between collocation points $q_i(0)$ (the power depends on the strength of the norm, the regularity of $\omega_0$, and the regularity of the RKHS defined by the GP $\curl\,\xi(x)$, see \cite{Wendland:2005,owhadi2019operator} for details and further references).

The second term $\s(x,t)$ is not well-understood, and we will seek to analyze it. Note that this term (1) is zero at the particle locations $q_i(t)$ ($\s(q_i(t),t)=0$), (2) a function of the  choice of kernel for $\xi$ and the number of particles $N$.
Although stability estimates\footnote{Stability estimates are available for $d=2$ \cite{zadrzynska2015stability}, they remain a challenge
for $d=3$ \cite{ladyzhenskaya2003sixth}.} for NS equations would allow us to bound the norm of the errors on velocity $u-\bar{u}$ and vorticity
$\omega-\bar{\omega}$, we do not expect those bounds to be useful since the chaotic nature of the NS equations would imply their rapid blow-up as a function of time (caused by a blow-up of the stability constants) in turbulent regimes.
On the other hand,  $\s(x,t)$ is a more useful measure of error since it acts as an instantaneous error made on the source term of the NS equations by the proposed numerical method: modulo the initial value error $\omega(x,0)-\bar{\omega}(x,0)$,
simulating $\bar{\omega}$ is equivalent to simulating the continuous NS equations with the added source term $\s(x,t)$.

\subsection{$\s$ as a measure of information loss}
$\s$ can also be interpreted as a measure of information loss. To describe this let $t_0\geq 0$ and let $q_a(t)$ be the trajectory of the particle driven by the flow $\bar{u}(x,t)$  ($\dot{q}_a(t)=\bar{u}\big(q_a(t),t\big)$) and
started at time $t_0$ at an arbitrary point $x\in \T^d$. Let $W_a(t)=\bar{\omega}\big(q_a(t),t\big)$ be the predicted vorticity at $q_a(t)$.
Let $q_e:=(q^\star,q_a)$ (resp. $W_e:=(W^\star,W_a)$) be the vector of particle locations obtained by concatenating $q^\star$ with $q_a$ (resp. $W^\star$ with $W_a$). Then the identity
\begin{equation}
u^\star(x,q^\star,W^\star)=u^\star(x,q_e,W_e)
\end{equation}
implies that
$(q_a,W_a)$ does not carry (additional) information on the approximation of the flow given the information contained in $(q^\star,W^\star)$.
Now, let $W_b$ be the solution of
 \begin{align}
{\bf (d=2)}\quad &\dot{W}_b(t)=\nu \Delta \curl\, u^\star\big(q_a(t),q^\star,W^\star\big)+g(q_a(t),t) \,,\label{eqY2dlo}\\
{\bf (d=3)}\quad &\dot{W}_b(t)=\nu \Delta \curl\, u^\star\big(q_e(t),q^\star,W^\star\big)+W_b(t) \nabla u^\star\big(q_e(t),q^\star,W^\star\big)+g(q_e(t),t) \label{eqY3dlo}\,,
\end{align}
with initial condition $W_b(t_0)=W_a(t_0)$, then the identity
\begin{equation}
W_b(t)-W_a(t)=\s(x,t) (t-t_0)+o(t-t_0)
\end{equation}
implies that $\big|\s(x,t)\big|$ can be interpreted as the instantaneous rate of information gain
 at time $t_0$ resulting from adding a particle at $x$ and letting $W_b$ be driven by the GPH equations. Similarly, $\s$
Equivalently  $\big|\s(x,t)\big|$ can be interpreted as the rate of information loss resulting from the absence of an additional particle at location $x$.
Therefore, to minimize information loss, the number of particles in GPH could be dynamically increased by adding new particles at locations $x$ where
$\big|\s(x,t)\big|$ is maximized (and a similar notion of information loss can be derived for removing particles).

\subsection{Turbulence as information loss}
How do you define and quantify turbulence? Observe that the current popular definition as
``the complex, chaotic motion of a fluid'' \cite{phillips2018turbulence} is not only empirical but also relative to the scale at which the flow is observed (the flow may appear laminar at fine scales and chaotic at coarse scales). From the GPH perspective, turbulence can be defined as the information loss incurred by approximating the dynamic of the continuous flow with the discrete information contained in $(q^\star,W^\star)$. In that sense, it is local quantity measured as $\big|\s(x,t)\big|$  and its definition relative to the information already contained $(q^\star,W^\star)$.

\begin{figure}[h]
    \centering
        \includegraphics[width=\textwidth ]{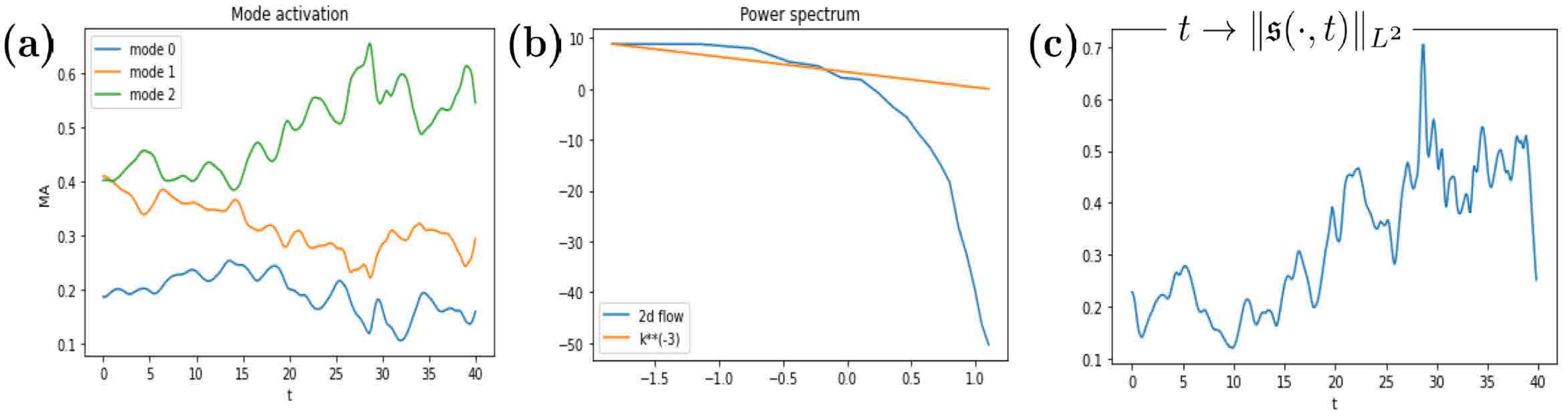}
    \caption{$\nu=0$. (a) Mode activation (b) Power spectrum (c) Source term error $t\rightarrow \|\s(\cdot,t)\|_{L^2}$.}
    \label{fiex1}
\end{figure}

\begin{figure}[h]
    \centering
        \includegraphics[width=\textwidth ]{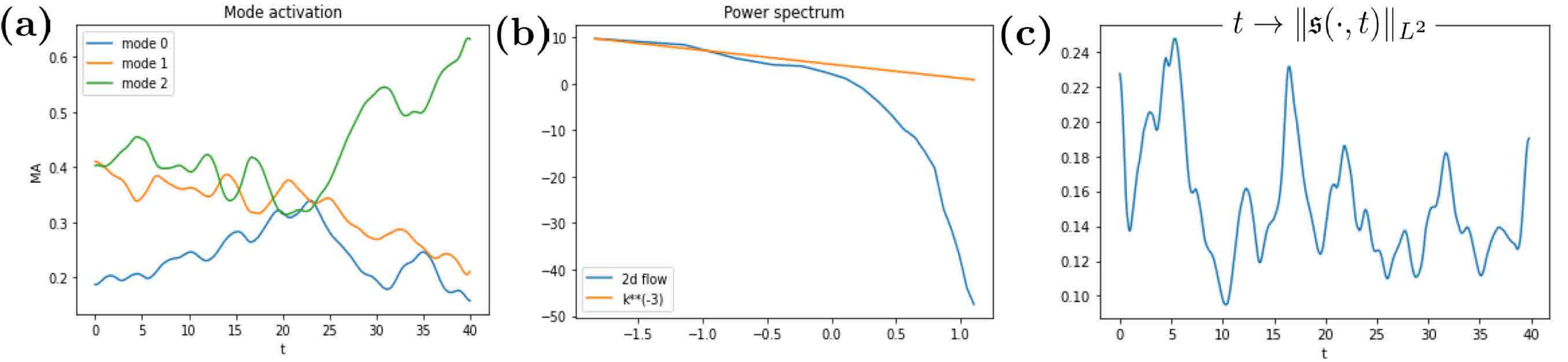}
    \caption{$\nu=0.001$. (a) Mode activation (b) Power spectrum (c) Source term error $t\rightarrow \|\s(\cdot,t)\|_{L^2}$.}
    \label{fiex1p2}
\end{figure}

\begin{figure}[h]
    \centering
        \includegraphics[width=\textwidth ]{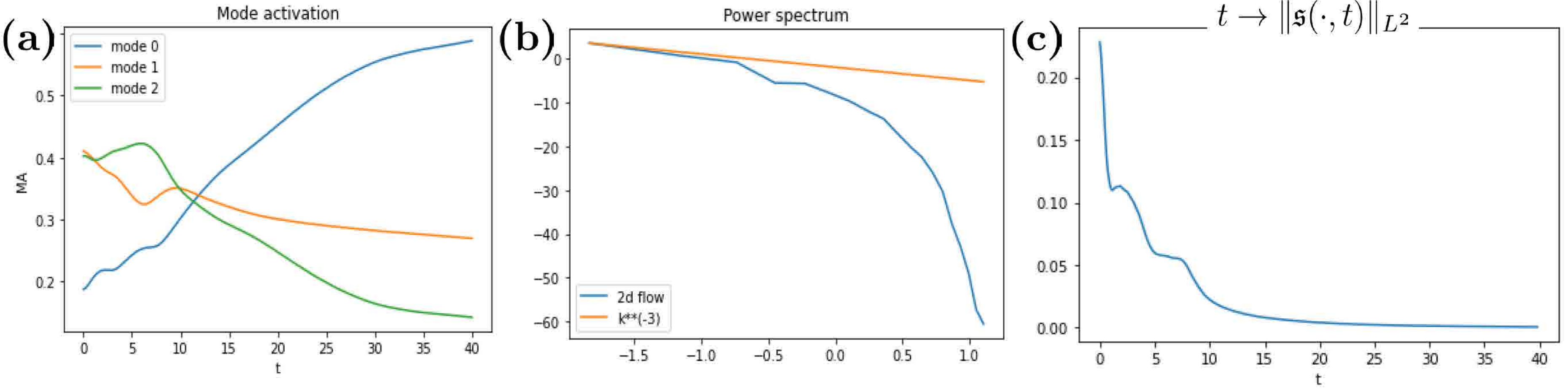}
    \caption{$\nu=0.01$. (a) Mode activation (b) Power spectrum (c) Source term error $t\rightarrow \|\s(\cdot,t)\|_{L^2}$.}
    \label{fiex1p3}
\end{figure}

\begin{figure}[h]
    \centering
        \includegraphics[width=\textwidth ]{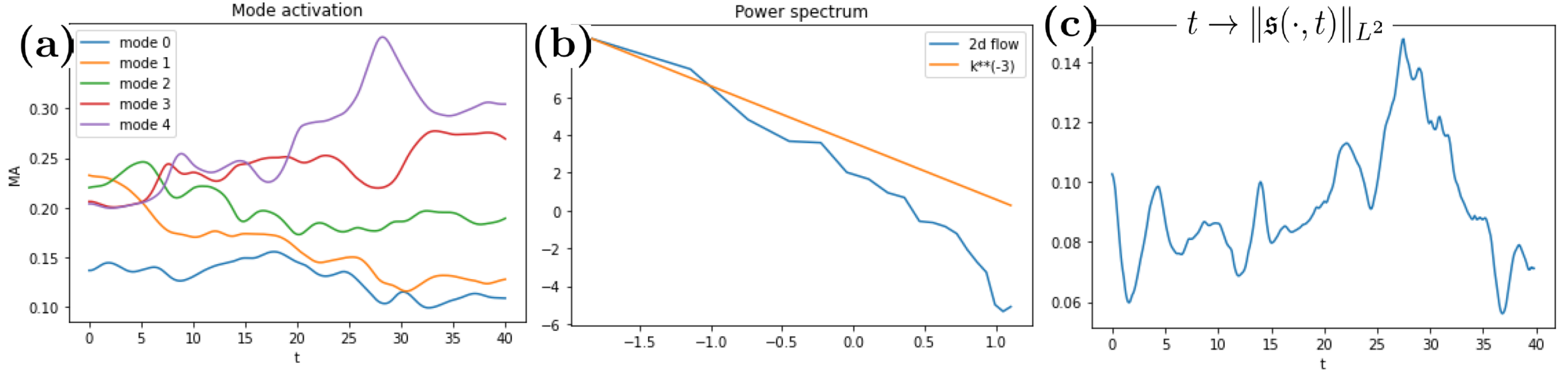}
    \caption{$\nu=0$. $5$ modes. (a) Mode activation (b) Power spectrum (c) Source term error $t\rightarrow \|\s(\cdot,t)\|_{L^2}$.}
    \label{fiex1m5}
\end{figure}

\section{Numerical experiments}\label{numericalexp}
In the following experiments we use, in dimension {\bf $d=2$}, the additive kernel of Sec.~\ref{secsecrcas} and \ref{secpwlaws} with $m+1=3$ modes, and $\gamma=4$, $(\sigma_0,\sigma_1,\sigma_2)=(2,1,0.5)$.
We use $N=25$ particles,  zero-forcing ($f=0$), zero viscosity ($\nu=0$), and initialize the vorticity field at random by sampling the initial value of  $W$ from the distribution of the Gaussian vector with the identity covariance matrix.
Fig.~\ref{fiex2} and \ref{fiex3} show snapshots of the velocity field ($x\rightarrow \bar{u}(x,t)$) and the vorticity field ($x\rightarrow \bar{\omega}(x,t)$)
with the entries of $q(t)$ shown as particles.
Fig.~\ref{fiex4} and \ref{fiex5} show snapshots of  the  variance of the velocity field ($x\rightarrow \Tr[\C^u(x,x)]$) and the variance of the vorticity field ($x\rightarrow \Tr[\C^\omega(x,x)]$).  
Fig.~\ref{fiex6} shows snapshots of  the  source terms error ($x\rightarrow \s(x,t)$).
Fig.~\ref{fiex1}.(a) shows the mode activation of each of the three modes as defined by \eqref{eqjkhkwhb1ddrn}.
Fig.~\ref{fiex1}.(b) shows the power spectrum of the field generated by our simulation and its comparison with the $k^{-3}$ power spectrum associated with  2d turbulence. Fig.~\ref{fiex1}.(c) shows the source term error $t\rightarrow \|\s(\cdot,t)\|_{L^2}$ where
$\|\s(\cdot,t)\|_{L^2}^2:=|\T^d|^{-1}\int_{\T^d}\s^2(x,t)\,dx$.
The plots shown in Fig.~\ref{fiex1} are for zero viscosity  $\nu=0$. Fig.~\ref{fiex1p2} and \ref{fiex1p3} show similar plots for $\nu=0.001$ and $\nu=0.01$.
Fig.~\ref{fiex1m5} shows similar plots for $\nu=0$, $m+1=5$ modes,  $\gamma=4$ and $(\sigma_0,\sigma_1,\sigma_2)=(2,1,1/2,1/4,1/8)$.
Note that compared to Fig.~\ref{fiex1}.(c), the source term error $t\rightarrow \|\s(\cdot,t)\|_{L^2}$ is decreased by one order of magnitude, which supports the point that our structured multiscale kernel leads to increased accuracy as the number of modes is increased.
\begin{table}[h]
\centering
\begin{tabular}{|c||*{5}{c|}}\hline
\backslashbox{${\bf m+1}$ }{${\bf \gamma}$ }
&\makebox[3em]{-2}&\makebox[3em]{0}
&\makebox[3em]{2}&\makebox[3em]{4}&\makebox[3em]{6}\\\hline\hline
2  & 27.6 & 39.4 & 93.4 & 31.1 & 34.5 \\\hline
3  & 0.118 & 0.158 & 0.214 & 0.389 & 1.144 \\\hline
4  & 1.14 & 4.22  & 0.076 & {\bf 0.070} & 0.266 \\\hline
\end{tabular}
\caption{Space/time average source term error $\|\s\|_{L^2}$ as a function  the number of modes $m+1$ and the power-law parameter $\gamma$.}
\label{tablesl2vsmg}
\end{table}

Table \ref{tablesl2vsmg} gives the space/time-averaged source term error ($T=40$)
\begin{equation}
\|\s\|_{L^2}:=\sqrt{T^{-1}\int_{\T^d\times [0,T]} \s^2(x,t) \,dx\,dt}
\end{equation}
 as a function of the number of modes ($m+1$ in the additive kernel of Sec.~\ref{secsecrcas} and \ref{secpwlaws}) and the value of the parameter $\gamma$ entering in the power law \eqref{eqjhedhhded}. For that table, we have used $N=100$ particles,  zero-forcing ($f=0$), $\nu=0.001$, and have initialized the vorticity field at random by sampling the initial value of  $W$ from the distribution of the Gaussian vector with identity covariance matrix. Choosing the parameter $\gamma$ close to the one ($\gamma=4$) matching the Kolmogorov scaling law and increasing the number of modes $m+1$ significantly diminishes the source term error $\|\s\|_{L^2}$.
 With only one mode ($m+1=1$, not shown in the table), the kernel is $K$ is too stiff to handle the transfer of energy towards fine scales, and accuracy significantly deteriorates ($\|\s\|_{L^2}\sim 17,000$,  and, without regularization with a nugget, velocity bursts are observed as particles come close to each other).
 
 \begin{Remark}
The values of $\|\s\|_{L^2}$ are absolute in table \ref{tablesl2vsmg}, and Fig.~\ref{fiex1p3}, \ref{fiex1m5}, \ref{fiex1p2}, \ref{fiex1},  \ref{fiex6}. Our main purpose is to show the dependence of $\|\s\|_{L^2}$ as a function of the number of modes and the value of the parameter $\gamma$. In particular, those values can be made relative by dividing them by  $17,000$ (the value of $\|\s\|_{L^2}$ with only one mode).
 \end{Remark}
 
\begin{figure}[h]
    \centering
        \includegraphics[width=0.4\textwidth ]{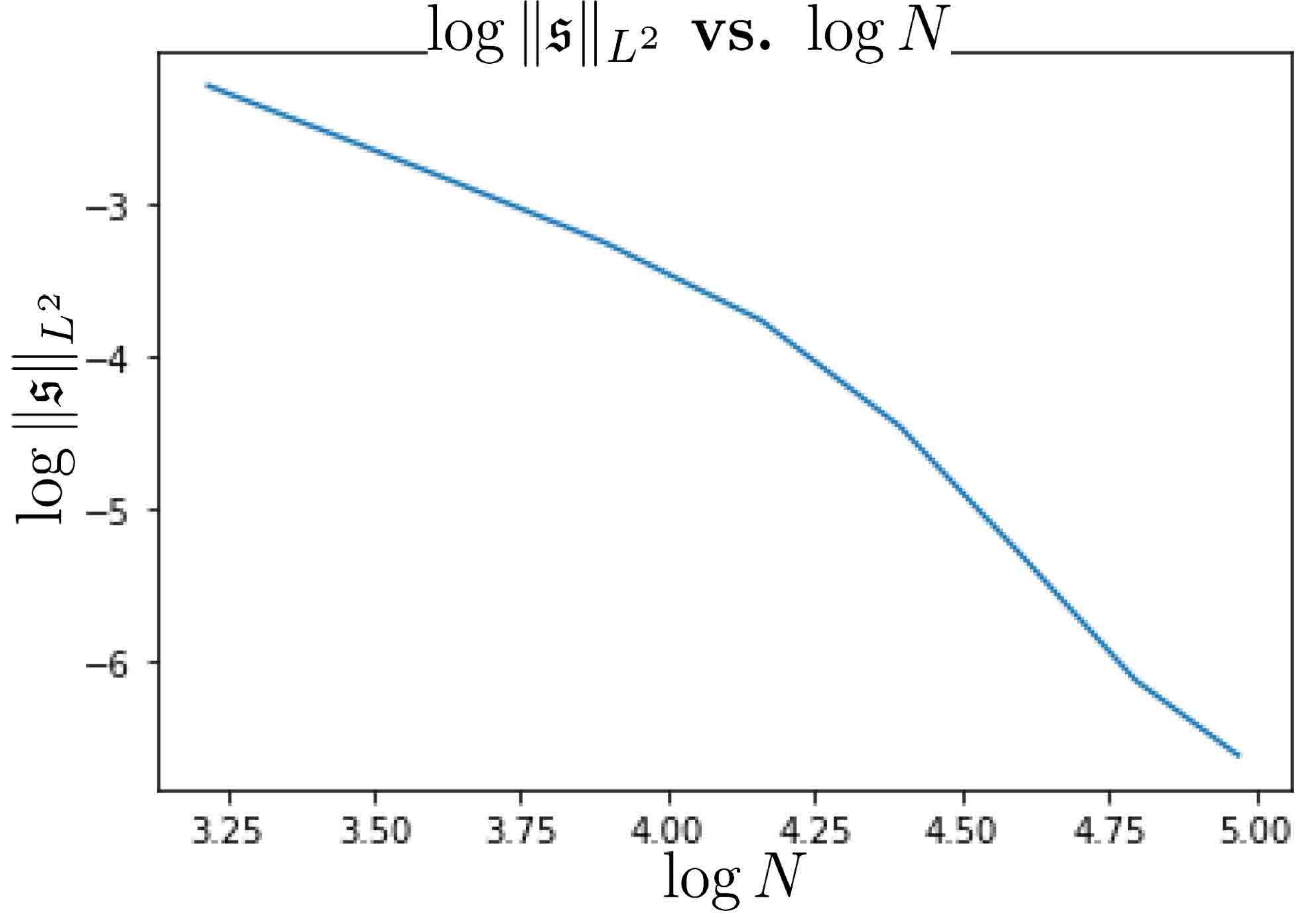}
    \caption{Error $\|\s\|_{L^2}$ vs. number of particles $N$ in log-log scale.}
    \label{figconvergence}
\end{figure}

 Fig.~\ref{figconvergence} illustrates the convergence of the method (as measured by $\|\s\|_{L^2}$) with respect to the number $N$ of particles. In that figure, with $W(0)=\omega_0(q(0))$, where the initial vorticity $\omega_0$ is chosen to be smooth and deterministic. The interpolation error in the approximation of the initial vorticity is not plotted (the analysis of this kernel interpolation error is classical \cite{owhadi2019operator}).

\begin{Remark}
The complexity of the method is proportional to the product between the number of time steps and the cost of inverting dense $N\times N$ kernel matrices. Although the sparse Cholesky factorization algorithms introduced in \cite{schafer2021sparse, SchaeferSullivanOwhadi17} could be adapted to potentially reduce the inversion cost to $\mathcal{O}(N \log^{2d}N)$, we have not employed this strategy here.
\end{Remark}
\begin{figure}[h]
    \centering
        \includegraphics[width=\textwidth ]{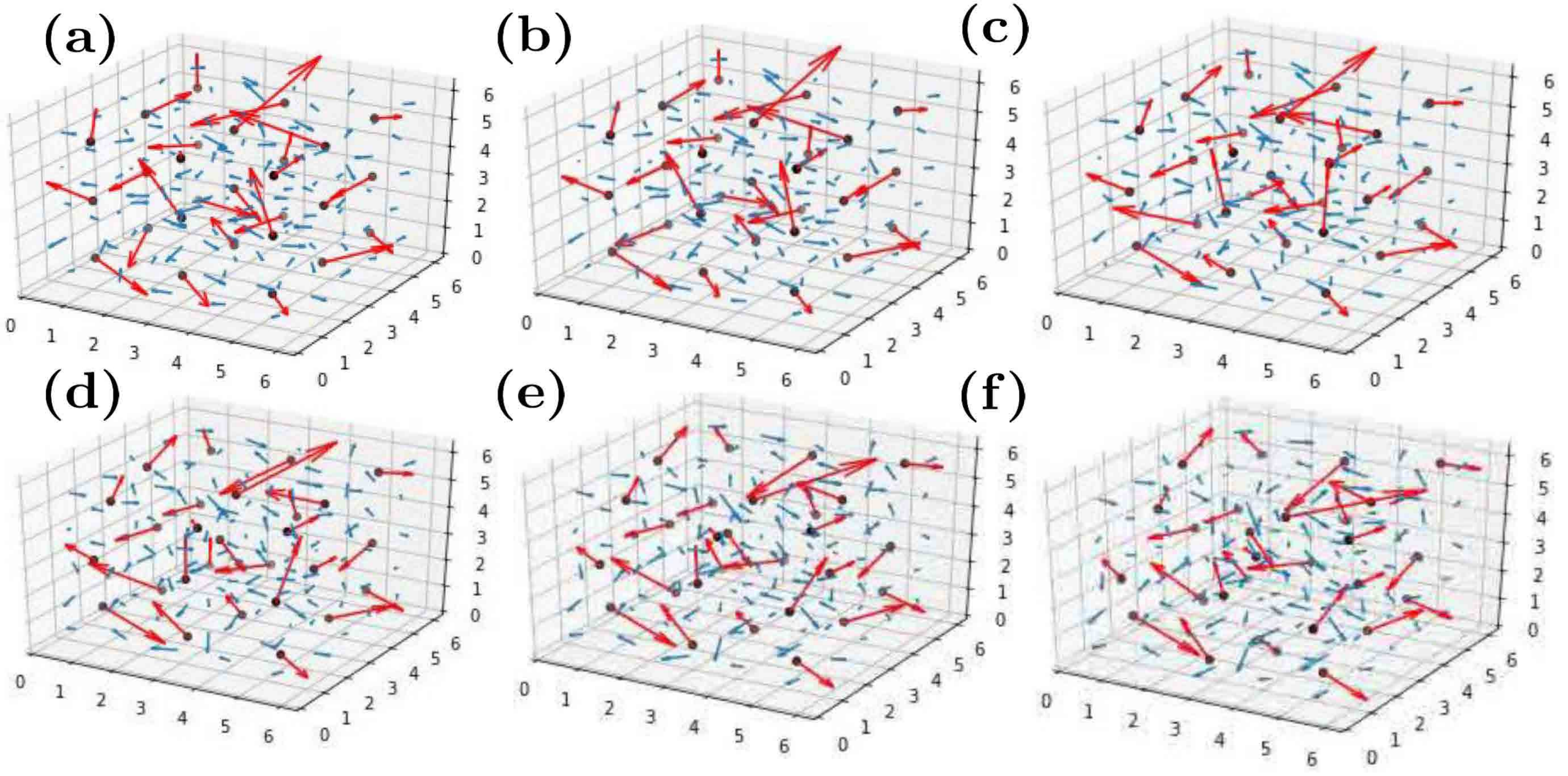}
    \caption{{\bf $d=3$}. Velocity and vorticity snapshots. Blue arrows show velocity, red arrows show vorticity $W$ at particle locations $q$. The kernel has $2$ modes.}
    \label{fiex3d2}
\end{figure}

\paragraph{{\bf The three-dimensional setting ($d=3$).}}
For the three-dimensional setting, we also use the additive kernel of Sec.~\ref{secsecrcas} and \ref{secpwlaws} with $m+1=1$ and $m+1=2$ modes,  $\gamma=2/3+2$, and $(\sigma_0,\sigma_1)=(2,1)$.
We use $N=9$ particles,  zero-forcing ($f=0$), non-zero viscosity ($\nu=0.001$), and initialize the vorticity field at random by sampling the initial value of  $W$ from the distribution of the Gaussian vector with the identity covariance matrix.
Fig.~\ref{fiex3d}  and \ref{fiex3d2} show snapshots of the velocity field ($x\rightarrow \bar{u}(x,t)$) and the vorticity $W$ at locations
 $q(t)$. Fig.~\ref{fiex3d} and Fig.~\ref{fiex3d2} employ one and two modes, respectively. The added mode increases the effective viscosity of the dynamic by acting as an energy sink.
 Compared to the two-dimensional setting, the three-dimensional ODE formulation of GPH has a quadratic term in $W^\star$ in \eqref{eqodesys2dt} that can lead to blowup in finite time. We do numerically observe this blowup and dampen the vortex stretching
 component of this quadratic term\footnote{Writing $w_i=w_{i,\parallel}+w_{i,\perp}$ for the orthogonal decomposition of
 $W^\star_i
\nabla_x \curl_x^T \L_y G(q_i^\star,q^\star) \big(\L_x \L_y G(q^\star,q^\star)\big)^{-1} W^\star$ into its projection along the direction of $W^\star_i$ and its orthogonal complement, we replace $w_i$ by $w_i=(1-\alpha) w_{i,\parallel}+w_{i,\perp}$.} by a factor $1-\alpha$ (with $\alpha\in [0,1)$) to avoid blow-up. Other strategies for avoiding a blowup in the numerical of the NS and Euler equations include numerical dissipation, and Lagrangian averaging \cite{marsden2003anisotropic}.
Although it is known that the three-dimensional Euler equations with boundary and smooth initial data can blow up, the blowup of the three-dimensional NS equations remains an open problem. Therefore addressing the possible blowup of \eqref{eqodesys2dt} in a manner that has better consistency with the underlying physics of turbulence remains an open problem that may require a modeling step (i.e., correcting the NS equations).
We also note that if the continuous three-dimensional NS equations do indeed blow up, then the solution obtained with GPH will exit the RKHS defined by a smooth kernel.
Finally, GPH may also serve as a potential candidate for identifying a singularity formation in the solution of the three-dimensional NS equations: if there exists a trajectory $(q^\star,W^\star)$ and a (possibly time-dependent) kernel $G$, such that
 $W^\star$ blows up in  finite time while
$\s$ in \eqref{eqvor3ds} (with $g=0$) remains smooth; then the NS equations do blow up in finite time
\cite{fefferman2000existence}.

\section{Further discussions}
\subsection{Choosing the kernel when the physics is unknown}\label{secchoker}

The approach proposed in this paper is to design the kernel to satisfy known physics. When the underlying physics is unknown, then the kernel can be learned from data via cross-validation/MLE in a given (possibly non-parametric) family of kernels
\cite{owhadi2019kernel, chen2020consistency, owhadi2020ideas}. The Kernel Flows  (a variant of cross-validation) approach \cite{owhadi2019kernel} has been shown to be efficient for learning (possibly stochastic) dynamical systems
\cite{hamzi2021learning, hamzi2021simple, darcy2021learning, lee2021learning, darcyone} and designing surrogate models \cite{prasanth2021kernel, susiluoto2021radiative, akian2022learning}. In particular, this approach has been shown to compare favorably to ANN-based methods (both in terms of complexity and accuracy) for weather/climate prediction using real satellite data \cite{hamzi2021simple}.

\subsection{GPH and ANN-based simulations}
The purpose of this manuscript is not to compare GPH against ANN-based methods for solving the NSE (we refer to \cite{chen2021solving} for such comparisons for general PDEs) but to highlight the fact that GP-based methods allow for incorporating the physics not solely through enforcing the PDE at a finite number of collocation points/particles but also through the choice and design of the kernel.
This being said, our analysis and results can be extended to derive an ANN variant of GPH. This variant can be obtained by simply defining the scalar-valued kernel $G$ introduced in Sec.~\ref{secdivfreegp} as
\begin{equation}\label{eqlwkldkdk}
G(x,x')=\psi_\theta^T(x) \psi_\theta(x')\,,
\end{equation}
where $\psi_\theta(x)$ is the output of an ANN, i.e., a function mapping $x$ to a finite-dimensional vector space parameterized by the parameters $\theta$ inner layers of a neural network. \eqref{eqlwkldkdk} then defines a parameterized kernel whose parameterized can be learned from data as described in Sec.~\ref{secchoker}.

\subsection{Uncertainty Quantification and Data Assimilation}
Since the proposed approach has a Bayesian interpretation, it naturally enables data assimilation and making predictions and estimations based on mixing simulation data with experimental data. To describe this assume that, in that to the information $(q,W)$ obtained from the simulation, we have access (as functions of time) to velocity measurements $v_1,\ldots,v_M$ at locations $z_1,\ldots,z_M$ (that may be time-dependent).
Then GPH can be modified to incorporate this information. To describe this, write,
\begin{equation}\label{eqswesehge}
u^\star\big(x,q,W, z, v\big):=\E\big[\xi(x)\big| \curl\, \xi(q)=W \text{ and } \xi(z)=v\big]\,.
\end{equation}
This modification can then be summarized as
approximating $u(x,t)$ with
\begin{equation}
\bar{u}(x,t):=u^\star(x,q^\star(t),W^\star(t), z, v)\,,
\end{equation}
and  $\omega(x,t)$ with
\begin{equation}
\bar{\omega}(x,t):=\curl\, u^\star(x,q^\star(t),W^\star(t), z,v)\,,
\end{equation}
where $(q^\star,W^\star)$ is the solution of the autonomous system of ODEs
 \begin{equation}\label{eqodesyedds}
  \begin{cases}
&\dot{q}^\star_i=u^\star\big(q_i^\star,q^\star,W^\star, z, v\big)\\
{\bf (d=2)}\quad &\dot{W}_i^\star(t)=\nu \Delta \curl\, u^\star\big(q_i^\star,q^\star,W^\star, z, v\big)+g(q_i^\star(t),t) \,,\\
{\bf (d=3)}\quad &\dot{W}_i^\star(t)=\nu \Delta \curl\, u^\star\big(q_i^\star,q^\star,W^\star, z, v\big)+W_i^\star \nabla u^\star\big(q_i^\star,q^\star,W^\star, z, v\big)+g(q_i^\star(t),t)\,,
\end{cases}
\end{equation}
with the initial condition $(q^\star,W^\star)(0)=(q,W)(0)=(q^0,\omega_0(q^0))$. Note that this modification is equivalent to
replacing the distribution of the GP $\xi$ in Sec.~\ref{secGPH} by that of a non-centered time dependent GP  with mean
$\E\big[\xi(x)\big|  \xi(z)=v\big]$ and covariance function defined as the conditional covariance of $\xi$ conditioned on $ \xi(z)=v$.
Representer formulas can naturally be obtained as in Sec.~\ref{secdivfreegp}. Other experimental measurements may be incorporated (e.g., vorticities at specific locations). Furthermore, using the proposed approach,  velocity and pressure fields can be learned from flow visualizations
as in \cite{raissi2020hidden}, with the advantage of also recovering uncertainties (whole posterior distributions) in addition to those fields.
To describe this, assume that one has access (as functions of time) to the values $y_1,\ldots,y_M$  at locations $z_1,\ldots,z_M$ of the concentration $c$ of a passive tracer satisfying the transport PDE
\begin{equation}
\partial_t c + u\cdot \nabla c=D \Delta c\,.
\end{equation}
Let $\Gamma$ be a smoothing scalar valued kernel and $\zeta \sim \cN(0,\Gamma)$. Write 
\begin{equation}
\bar{c}(x,t):=\E[\zeta(x)|\zeta(z)=y(t)]\,,
\end{equation}
and
\begin{equation}\label{eqswesehgeb}
u^\star\big(x,q,W, t\big):=\E\big[\xi(x)\big| \curl\, \xi(q)=W \text{ and } \partial_t \bar{c}(z,t) + \xi(z)\cdot \nabla \bar{c}(z,t)=D \Delta \bar{c}(z,t)\big]\,.
\end{equation}
$u(x,t)$ can then be approximated with
\begin{equation}
\bar{u}(x,t):=u^\star(x,q^\star(t),W^\star(t), t)\,,
\end{equation}
and  $\omega(x,t)$ with
\begin{equation}
\bar{\omega}(x,t):=\curl\, u^\star(x,q^\star(t),W^\star(t), t)\,,
\end{equation}
where $(q^\star,W^\star)$ is the solution of the autonomous system of ODEs
 \begin{equation}\label{eqodesededds}
  \begin{cases}
&\dot{q}^\star_i=u^\star\big(q_i^\star,q^\star,W^\star, t\big)\\
{\bf (d=2)}\quad &\dot{W}_i^\star(t)=\nu \Delta \curl\, u^\star\big(q_i^\star,q^\star,W^\star, t\big)+g(q_i^\star(t),t) \,,\\
{\bf (d=3)}\quad &\dot{W}_i^\star(t)=\nu \Delta \curl\, u^\star\big(q_i^\star,q^\star,W^\star, t\big)+W_i^\star \nabla u^\star\big(q_i^\star,q^\star,W^\star, t\big)+g(q_i^\star(t),t)\,,
\end{cases}
\end{equation}
with the initial condition $q^\star(0)=q^0$ and 
\begin{equation}
W^\star(0)=\E\big[\curl \xi(q^0)\big| \partial_t \bar{c}(z,0) + \xi(z)\cdot \nabla \bar{c}(z,0)=D \Delta \bar{c}(z,0)\big]\,.
\end{equation}

\subsection*{Acknowledgments}
The author gratefully acknowledges partial support from the Air Force Office of Scientific Research under MURI award number FA9550-20-1-0358 (Machine Learning and Physics-Based Modeling and Simulation) and from the Department of Energy under award number DE-SC0023163 (SEA-CROGS: Scalable, Efficient and Accelerated Causal Reasoning Operators, Graphs and Spikes for
Earth and Embedded Systems). The author also thanks two anonymous referees for comments and suggestions.

\bibliographystyle{plain}
\bibliography{merged,RPS,extra,kmd}

\end{document}